\documentclass[a4paper,USenglish,cleveref,autoref,thm-restate]{lipics-v2021}

\usepackage[algo2e,algoruled,vlined,linesnumbered]{algorithm2e}
\usepackage{setspace}
\usepackage{xspace}

\SetKwInput{Input}{Input}
\SetKwInput{Output}{Output}
\SetKwInput{Ret}{Return}
\SetKw{Proc}{Procedure}
\SetKw{out}{output}
\SetKwFunction{Enum}{enumerate}
\SetKwFunction{Ext}{extendable}
\newcommand*{\halt}{\textsc{halt}\xspace} 

\let\oldnl\nl
\newcommand{\nonl}{\renewcommand{\nl}{\let\nl\oldnl}}

\crefname{prop}{property}{properties}
\creflabelformat{prop}{#2{(\textit{#1})}#3}
\crefname{cond}{condition}{conditions}

\newcommand*{\nwspace}{\hspace*{.1em}} 

\newcommand*{\Fyp}{\mathcal{F}}
\newcommand*{\Gyp}{\mathcal{G}}
\newcommand*{\Hyp}{\mathcal{H}}

\newcommand*{\TransRank}{\textsc{Transversal Rank}\xspace}
\newcommand*{\ConfDeg}{\textsc{Conformal Degree}\xspace}
\newcommand*{\ExtHS}{\textsc{Extension MinHS}\xspace}

\newcommand*{\poly}{\textup{\textsf{poly}}}

\DeclareMathOperator{\C}{\mathcal{C}}

\DeclareMathOperator{\rank}{rank}
\DeclareMathOperator{\Tr}{Tr}

\newcommand*{\FPT}{\textsf{FPT}\xspace}
\newcommand*{\NP}{\textsf{NP}\xspace}
\newcommand*{\W}{\textsf{W}\xspace}
\newcommand{\co}{\textsf{co}}
\newcommand{\para}{\textsf{para}}

\let\oldsqrt\sqrt
\def\hksqrt{\mathpalette\DHLhksqrt}
\def\DHLhksqrt#1#2{\setbox0=\hbox{$#1\oldsqrt{#2\,}$}\dimen0=\ht0
   \advance\dimen0-0.2\ht0
   \setbox2=\hbox{\vrule height\ht0 depth -\dimen0}%
   {\box0\lower0.4pt\box2}}
\renewcommand\sqrt\hksqrt

\renewcommand*{\le}{\leqslant}
\renewcommand*{\ge}{\geqslant}

\providecommand{\ignore}[1]{}

\title{Transversal Rank, Conformality and Enumeration}

\author{Martin Schirneck}{Karlsruhe Institute of Technology}{martin.schirneck@kit.edu}{https://orcid.org/0000-0001-7086-5577}{The author is supported by the German Research Foundation (DFG)
under grant agrmt.\ No.~556899211 (``Design, Analysis, and Engineering of Enumeration Algorithms'').}

\authorrunning{M.~Schirneck}
\Copyright{Martin Schirneck} 



\category{} 

\relatedversion{} 

\nolinenumbers 
\hideLIPIcs

\EventEditors{Sayan Bhattacharya, Michael Benedikt, and Danupon Nanongkai}
\EventNoEds{3}
\EventLongTitle{53rd International Colloquium on Automata, Languages, and Programming (ICALP 2026) }
\EventShortTitle{ICALP 2026}
\EventAcronym{ICALP}
\EventYear{2026}
\EventDate{July 7--10, 2026}
\EventLocation{London, United Kingdom}
\EventLogo{}
\SeriesVolume{X}
\ArticleNo{X}

\begin{document}

\maketitle

\begin{abstract}
	The \emph{transversal rank} of a hypergraph is the maximum size of its minimal hitting sets.
	Deciding, for an $n$-vertex, $m$-edge hypergraph and an integer $k$,
	whether the transversal rank is at least $k$ takes time $O(m^{k+1} \nwspace n)$
	with an algorithm that is known since the 70s.
	It essentially matches an $(m+n)^{\Omega(k)}$ ETH-lower bound
	by Araújo, Bougeret, Campos, and Sau [Algorithmica 2023] and Dublois, Lampis, and Paschos~[TCS 2022].
	Many hypergraphs seen in practice have much more edges than vertices, $m \gg n$.
	This raises the question whether an improvement of the run time dependency on $m$
	can be traded for an increase in the dependency on $n$.
	
	Our first result is an algorithm to recognize hypergraphs with transversal rank at least $k$
	in time $O(\Delta^{k-2} \nwspace mn^{k-1})$, where $\Delta \le m$ is the maximum degree.
	Our main technical contribution is a ``look-ahead'' method
	that allows us to find \emph{higher-order extensions},
	minimal hitting sets that augment a given set with at least two more vertices.
	We show that this method can also be used to \emph{enumerate}
	all minimal hitting sets of a hypergraph with transversal rank $k^*$
	with delay $O(\Delta^{k^*-1} \nwspace mn^2)$.
	
	We then explore the possibility of further reducing the running time 
	for computing the transversal rank 
	to $\poly(m) \cdot n^{k+O(1)}$.
 	This turns out to be \emph{equivalent} to several breakthroughs
 	in combinatorial algorithms and enumeration.
	Among other things, such an improvement is possible if and only if $k$-conformal hypergraphs
	can also be recognized in time $\poly(m) \cdot n^{k+O(1)}$,
	and iff the maximal hypercliques/independent sets
	of a uniform hypergraph can be enumerated with incremental delay.
\end{abstract}

\ignore{%
	The \emph{transversal rank} of a hypergraph is the maximum size of its minimal hitting sets.
	Deciding, for an $n$-vertex, $m$-edge hypergraph and an integer $k$,
	whether the transversal rank is at least $k$ takes time $O(m^{k+1} n)$
	with an algorithm that is known since the 70s.
	It essentially matches an $(m+n)^{\Omega(k)}$ ETH-lower bound
	by Araújo, Bougeret, Campos, and Sau [Algorithmica 2023] and Dublois, Lampis, and Paschos [TCS 2022].
	Many hypergraphs seen in practice have much more edges than vertices, $m \gg n$.
	This raises the question whether an improvement of the run time dependency on $m$
	can be traded for an increase in the dependency on $n$.
	
	Our first result is an algorithm to recognize hypergraphs with transversal rank at least $k$
	in time $O(\Delta^{k-2} mn^{k-1})$, where $\Delta \le m$ is the maximum degree.
	Our main technical contribution is a ``look-ahead'' method
	that allows us to find \emph{higher-order extensions},
	minimal hitting sets that augment a given set with at least two more vertices.
	We show that this method can also be used to \emph{enumerate}
	all minimal hitting sets of a hypergraph with transversal rank $k^*$
	with delay $O(\Delta^{k^*-1} mn^2)$.
	
	We then explore the possibility of further reducing the running time 
	for computing the transversal rank 
	to $\textsf{poly}(m) \cdot n^{k+O(1)}$.
 	This turns out to be \emph{equivalent} to several breakthroughs
 	in combinatorial algorithms and enumeration.
	Among other things, such an improvement is possible if and only if $k$-conformal hypergraphs
	can also be recognized in time $\textsf{poly}(m) \cdot n^{k+O(1)}$,
	and iff the maximal hypercliques/independent sets
	of a uniform hypergraph can be enumerated with incremental delay.
}

\section{Introduction}
\label{sec:intro}


A plethora of results in graph algorithms discuss methods to decide
whether a given input belongs to a certain graph class.
We have algorithms that recognize planar graphs, graphs of treewidth at most 67,
or graphs with independence number at least 42.
For hypergraphs, however, much fewer methods are known
and most of the literature focuses on uniform hypergraphs,
where all hyperedges have the same number of vertices \cite{Abboud18MoreConsecuences,DorobiszKozik23HypergraphColoring,Fomin10IterativeCompression}, 
or, equivalently, binary matrices in which every column has the same number of 1s~\cite{Cooper19RankRandomMatrix}.
Our goal is to advance the area of recognition algorithms for classes of \emph{non-uniform} hypergraphs.
Specifically, we are interested in hypergraphs with large \emph{transversal rank},
which is the \emph{maximum} size of any \emph{inclusion-wise minimal} hitting set.\footnote{
	A \emph{transversal} is another word for hitting set.
	The \emph{transversal hypergraph} is the collection of all minimal hitting sets
	and the transversal rank is the maximum edge size of this hypergraph.
}
While our target are non-uniform hypergraphs,
our analysis leverages ideas from the uniform realm.
It also reveals deep connections to the problem of enumerating all minimal hitting sets.

The transversal rank was introduced to gauge the worst-case performance of ``local improvement'' methods for the hitting set problem. 
It found applications in areas like parameterized algorithms and kernelization~\cite{Hols22BlockingSets,Zehavi17MaximumMinimalVertexCover},
formal concept analysis~\cite{ColombNourine08KeysofFormalContext},
and computational biology~\cite{Damaschke11DoubleHypergraphDualization}.
Let \TransRank denote the problem of deciding,
for an $n$-vertex, $m$-edge hypergraph and a non-negative integer $k$,
whether there is a minimal hitting set with at least $k$ vertices.
Fittingly, the problem is also known under the name \textsc{Maximum Minimal Hitting Set}
\cite{Damaschke11DoubleHypergraphDualization}.
It generalizes many other ``max-min'' variants of classical optimization problems
including \textsc{Maximum Minimal Vertex Cover}~\cite{Boria15MaxMinVertexCover},
\textsc{Maximum Minimal Feedback Vertex Set}~\cite{Dublois22MaximumMinimalFVS},
\textsc{Maximum Minimal Blocking Set} (where all maximum independent sets of a graph need to be hit~\cite{Araujo23MaximumMinimalBlockingSetAlgorithmica}),
as well as \textsc{Upper Domination} (maximum minimal dominating set)
and its variant \textsc{Upper Total Domination}~\cite{Bazgan18FacetsofUpperDomination}.

The max-min variants are often harder then their underlying minimization problems.
Known lower bounds on the computational complexity of the \textsc{Upper Domination} problem
immediately transfer to \TransRank.
A result by Cheston, Fricke, Hedetniemi, and Jacobs~\cite{Cheston90UpperFractionalDomination}
thus shows that the problem is \NP-complete.
Bazgan, Brankovic, Casel, Fernau, Jansen, Klein, Lampis, Liedloff, Monnot, 
and Paschos~\cite{Bazgan18FacetsofUpperDomination}
proved that it is $\W[1]$-hard when parameterized by $k$
and that even approximating the transversal rank
within a factor of $n^{1-\varepsilon}$ is \NP-hard.
They further showed that
there is no algorithm for \TransRank running in time $(m + n)^{o(\sqrt{k} \nwspace)}$,
unless the Exponential Time Hypothesis (ETH) fails.
Recently, two independent works by Araújo, Bougeret, Campos, and Sau~\cite{Araujo23MaximumMinimalBlockingSetAlgorithmica} as well as
Dublois, Lampis, and Paschos~\cite{Dublois22UpperDominatingSetTCS}
improved this to $(m+n)^{o(k)}$.

Currently, the most efficient algorithm to recognize hypergraphs with large minimal hitting sets
follows from a result by Berge and Duchet~\cite{BergeDuchet75GilmoresTheorem} from the 70s.
Here, a \emph{Sperner} hypergraph is one in which no two edges are contained in one another.\footnote{
	Sperner hypergraphs are sometime also called \emph{simple}, e.g.\ in \cite{Berge89Hypergraphs},
	but this term is not used consistently in the literature, see~\cite{Frieze13Simple}.
}

\begin{lemma}[Berge and Duchet~\cite{BergeDuchet75GilmoresTheorem};
see Corollary~1, p.~58 of~\cite{Berge89Hypergraphs}]
\label{lem:BergeDuchet_subhypergraph}
	Let $\Hyp$ be a Sperner hypergraph and $k \ge 3$ an integer.
	Hypergraph $\Hyp$ has transversal rank at least $k$ if and only if
	there exists $k$ distinct edges $E_1, E_2, \dots, E_k \in \Hyp$
	such that every edge $E \in \Hyp$ has a vertex $x \in E$
	that appears in at most one of the $E_i$.
	Moreover, if $D$ denotes the set of all such $x$,
	then \emph{every} minimal hitting set $T \subseteq D$ has size at least $k$.
\end{lemma}

For hitting sets, it does not make a difference whether the full hypergraph $\Hyp$
or only its inclusion-wise minimal edges are considered.
Finding those takes time $O(m^2 \nwspace n)$ and the resulting subhypergraph is Sperner,
allowing an application of \autoref{lem:BergeDuchet_subhypergraph}.
Given a collection $\{E_i\}_i$ of $k$ edges in $\Hyp$, 
checking whether the resulting set $D$ is a hitting set for $\Hyp$ can be done in time $O(mn)$.
Cycling through all $k$-subhypergraphs of $\Hyp$ thus gives a total time of $O(m^{k+1} \nwspace n)$.
We show in \autoref{app:intro_proofs}
how a more careful preprocessing can reduce the running time to
$O(m^{k+1} + m^k \nwspace n)$, which is never worse than $O(m^{k+1} \nwspace n)$.
Producing an actual witness for the large transversal rank,
namely, a \emph{minimal} hitting set of size at least $k$,
hardly takes any additional effort. 
Once a suitable set $D$ is found,
an arbitrary minimal solution $T \subseteq D$ can be computed in time $O(mn)$
(see also \autoref{app:intro_proofs}).

The algorithm matches the lower bound by Araújo et al.~\cite{Araujo23MaximumMinimalBlockingSetAlgorithmica}
as well as Dublois, Lampis, and Paschos~\cite{Dublois22UpperDominatingSetTCS}
up to constant factors in the exponent.
However, its large complexity in terms of $m$ makes it prohibitively
expensive for hypergraphs with many edges,
which is the typical setting in many practical applications~\cite{Benson18HigherOrderLinkPrediction,Birnick20HPIValid,Papenbrock15SevenAlgorithms,Stavropoulos16FrequentItemsetHiding}.
Note that in non-uniform hypergraphs (in constrast to graphs)
$m$ is not polynomially bounded in $n$.
ETH excludes the possibility that the exponents of $m$ and $n$
in the running time can be reduced to $o(k)$ \emph{simultaneously}.
Notwithstanding, the fine-grained trade-offs between the two
are entirely unexplored.
Running times of, say, $m^{k-10} \cdot \poly(n)$ and $O(m \cdot n^{3k/2})$
are both consistent with our current knowledge,
but would result in drastically different behavior in practice.

We investigate whether one can reduce the dependency on the number of edges $m$,
even if this may increase the complexity in terms of the number of vertices $n$.
Our first result is a faster algorithm for hypergraphs with very many edges,
or instances in which every vertex has only small degree.
Let $\Delta \le m$ denote the maximum degree.

\begin{restatable}{theorem}{transversalrankalg}
\label{thm:transversal_rank_algorithm}
	Hypergraphs with transversal rank at least $k$ can be recognized 
	in time $O(\Delta^{k-2} \nwspace mn^{k-1})$.
	For $k < 2$, the running time is $O(mn)$.
\end{restatable}

\noindent
When comparing \autoref{lem:BergeDuchet_subhypergraph} and
\autoref{thm:transversal_rank_algorithm}, the dependency on $m$ drops
from $m^{k+1}$ to $\Delta^{k-2} \nwspace m$,
which, however, increases the running time in terms of the number of vertices to $n^{k-1}$.
The trade-off is beneficial whenever $m  = \Omega( \nwspace (\Delta n)^{\frac{k-2}{k}})$.

Consider the extension problem for minimal hitting sets (\ExtHS),
where we have to check whether a given set of vertices $X$ is contained in a minimal hitting set.
Boros, Gurvich, and Hammer~\cite{Boros98Subimplicants} showed that the problem is \NP-complete,
and Bläsius, Friedrich, Lischeid, Meeks, and Schirneck~\cite{Blaesius22EfficientlyJCSS}
proved that it is $\W[3]$-complete when parameterized by the cardinality $|X|$ of the partial solution.
We investigate the variant that decides
whether there exists a minimal hitting set that contains $X$ 
and uses \emph{at least two more} vertices.
We call this a \emph{higher-order extension}.
The central contribution of the first part of our work is an algorithm 
that allows us to \emph{look ahead} in the search for higher-order extensions of $X$.
It runs in the same time as the best-known method for the original extension problem.
This is the key ingredient to establish \autoref{thm:transversal_rank_algorithm}.
Along the way, we identify the higher-order variant as the hard core of the extension problem.
This immediately shows that it is also \NP- and $\W[3]$-complete and implies several fine-grained lower bounds.
(See \autoref{sec:overview} for details.)

The look-ahead idea also has applications in the \emph{enumeration} of minimal hitting sets.
Finding a \emph{minimum-cardinality} hitting set is classically \NP-hard~\cite{Karp72Reducibility},
but a simple greedy algorithm can produce a minimal one.
Between the two extremes is the task of enumeration, that is, computing and listing
\emph{all} minimal hitting sets of a hypergraph.
There can be exponentially many minimal solutions, 
so a polynomial algorithm is outright impossible.
The hope would be to achieve \emph{output-polynomial} time,
scaling polynomially with the input size and the number of solutions.
It is \emph{the} major open problem in enumeration
whether the so-called \emph{transversal hypergraph problem}
admits an output-polynomial algorithm~\cite{CapelliStrozecki19Incremental,DemetrovicsThi87Antikeys,Mannila87Dependency,Reiter87DiagnosisFirstPrinciples}.

While answering this question seems to be out of reach of our current techniques,
a lot of work has been dedicated to identify classes of hypergraphs for which this problem
is tractable, see e.g.~\cite{EiterGottlob95RelatedProblems,EiterGottlobMakino03NewResults,%
EiterMakinoGottlob08Survey,Khachiyan07BoundedEdgeIntersection}.
For hypergraphs for which the transversal rank $k^*$ is bounded,
even better than mere output-polynomial methods are possible.
The \emph{delay} is the time between consecutive outputs.
This includes the preprocessing before the first solution and
the time after the last one until termination.
Bläsius et al.~\cite{Blaesius22EfficientlyJCSS} used the \ExtHS problem
to devise an algorithm with delay $O(m^{k^*+1} \nwspace n^2)$,
The crucial property of the algorithm is
that it does \emph{not} need to know $k^*$ in advance,\footnote{
	If $k^*$ were known and indeed a constant, 
	merely brute-forcing all sets up to size $k^*$
	would give a trivial polynomial (thus output-polynomial) algorithm.
}
only the analysis depends on $k^*$.
Araújo et al.~\cite{Araujo23MaximumMinimalBlockingSetAlgorithmica}
later observed that the bound can be tightened to $O(\Delta^{k^*} mn^2)$,
which is more efficient if the degrees are small,
but gives the worst case running time if $\Delta = \Omega(m)$.
Our improvement regarding the higher-order extensions allows
us to break the $\Omega(m^{k^*+1}) \cdot \poly(n)$ delay barrier
for general hypergraphs.

\begin{restatable}{theorem}{improveddelay}
\label{thm:improved_delay}
	The minimal hitting sets of a hypergraph with transversal rank $k^*$ can be enumerated
	with delay $O(\Delta^{k^*-1} \nwspace mn^2)$.
	The algorithm uses space $O(mn)$.
\end{restatable}

The second part of this work explores the question whether the dependency on $m$
in the running time to compute the transversal rank can be reduced beyond \autoref{thm:transversal_rank_algorithm}.
Ideally, we would like it to be polynomial in $m$,
with an exponent that is independent of $k$.
Of course, we cannot expect to achieve complexity $n^{o(k)}$ at the same time.
Instead, we investigate the consequences of a method with running time $\poly(m) \cdot n^{k+O(1)}$.
As it turns out, such an improvement
is \emph{equivalent} to several other breakthroughs in combinatorial algorithms 
for conformal hypergraphs as well as the enumeration of maximal hypercliques.

Consider a set of vertices $S$ for which there exists a hyperedge $E \,{\in}\, \Hyp$ 
with $S \,{\subseteq}\, E$.
Evidently, each $k$-subset of $S$ is also contained in an edge of $\Hyp$.
The hypergraph is called $k$-\emph{conformal} if the opposite is true,
that is, all $k$-subsets of $S$ being contained in some edge 
is already sufficient for the existence of an edge that contains all of $S$.
The \emph{conformal degree} of $\Hyp$ is the smallest $k$ for which $\Hyp$ is $k$-conformal.
Already Berge and Duchet~\cite{BergeDuchet75GilmoresTheorem} observed
that there is a close connection between hypergraphs with large transversal rank
and large conformal degree (see \autoref{lem:BergeDuchet_conformal}).
It is thus not surprising that recognizing $k$-conformal hypergraphs is \co\NP-complete,
and $\co\W[1]$-hard when parameterized by $k$
\cite{Bazgan18FacetsofUpperDomination,ColombNourine08KeysofFormalContext}.
A \emph{hyperclique} in a $k$-uniform hypergraph $\Gyp$ is a set $S$
such that each $k$-subset of $S$ is \emph{equal} to an edge of $\Gyp$.

We extend the connection between transversal rank and conformal degree to enumeration.
We say the solutions to a combinatorial problem can be listed with \emph{incremental delay}
if the $i$-th delay, the time between the $i$-th and $(i{+}1)$-st output,
is bounded by a polynomial of the input size and $i$.
We reserve the variable $i$ for the number of solutions seen so far throughout this work.
In general, the notion of incremental delay lies between output-polynomial time and polynomial delay,
but for hitting sets
it is known that any output-polynomial algorithm can be turned
into one with incremental delay \cite{Bioch95Identification}.
This includes hypergraphs whose \emph{rank}  $r$ (the maximum edge size) is bounded by a constant.
However, like for the transversal rank,
a fine-grained discussion of the $\poly(m,n,i)$-delay is missing from the literature.
A method with delay $i^3 \cdot n^{r+O(1)}$ would be preferred over
one with delay $O(i^{r+2} \cdot n)$, 
the latter being the current state of the art~\cite{EiterGottlob95RelatedProblems}.
We prove a quantitative equivalence between the computation of the conformal degree,
the enumeration of minimal hitting sets in rank-bounded hypergraphs,
and the enumeration of maximal hypercliques.
Of course, everything that is said about hypercliques also applies to independent sets
by swapping the role of edges and non-edges.

\begin{restatable}{theorem}{mainequivalence}
\label{thm:main_equivalence}
	The following statements are equivalent.
	\vspace*{.25em}
	\begin{itemize}
		\item Hypergraphs with transversal rank at least $k$ can be recognized 
			in time $\poly(m) \cdot n^{k+O(1)}$.
		\vspace*{.25em}
		\item Hypergraphs that are $k$-conformal can be recognized
			in time $\poly(m) \cdot n^{k+O(1)}$.				
		\vspace*{.25em}
		\item The minimal hitting sets of hypergraphs with rank $r$ can be enumerated
			with incremental delay $\poly(i) \cdot n^{r+O(1)}$.
		\vspace*{.25em}
		\item The maximal hypercliques/independent sets of uniform hypergraphs with rank $r$ can be enumerated
			with incremental delay $\poly(i) \cdot n^{r+O(1)}$.
	\end{itemize}
	\vspace*{.5em}
\end{restatable}

On the one hand, the equivalence shows how faster enumeration algorithms in uniform hypergraphs
can help with decision problems also for non-uniform inputs.
On the other hand, it opens new inroads to prove (conditional) delay lower bounds,
a notoriously hard problem in enumeration complexity~\cite{CapelliStrozecki19Incremental}.
Ruling out fast listing algorithms for minimal hitting sets/maximal hypercliques in uniform hypergraphs
now reduces to proving lower bounds on the \TransRank problem.
We point out that the hypergraphs in the first two items of \autoref{thm:main_equivalence}
do \emph{not} need to have bounded rank.
Indeed, if $r$ is constant, Araújo et al.~\cite{Araujo23MaximumMinimalBlockingSetAlgorithmica} 
showed how to compute the transversal rank much faster
in time $2^{rk} \cdot \poly(m,n)$.

Traditionally, the enumeration of minimal hitting sets has been investigated
using a different kind of decision problem, namely,
given two hypergraphs $\Hyp$ and $\Gyp$,
verifying whether $\Gyp$ consists precisely of the minimal hitting sets of $\Hyp$.
This problem is known to be solvable in polynomial time (in the combined input size of $\Hyp$ and $\Gyp$)
if and only if the minimal hitting sets can be enumerated
with incremental delay~\cite{Bioch95Identification}.
We strengthen this equivalence in the case of rank-$r$ hypergraphs $\Hyp$.

\begin{restatable}{theorem}{secondequivalence}
\label{thm:second_equivalence}
	Let $\Hyp$ be a hypergraph with rank $r$.
	The following statements are equivalent.
	\vspace*{.25em}
	\begin{itemize}
		\item The minimal hitting sets of $\Hyp$ can be enumerated
			with incremental delay $\poly(i) \cdot  n^{r+O(1)}$.\vspace*{.25em}
		\item For any hypergraph $\Gyp$, deciding whether it contains all and only the minimal hitting sets of $\Hyp$ can be done in time $\poly(|\Gyp|) \cdot n^{r+O(1)}$.
	\end{itemize}
	\vspace*{.5em}
\end{restatable}

The heart of \autoref{thm:main_equivalence} is to show
that recognizing hypergraphs with large transversal rank
is the same as quickly enumerating the maximal hypercliques 
of a uniform hypergraph with suitably chosen rank.
For simple graphs (2-uniform hypergraphs),
there are non-trivial clique enumeration algorithms.
Tsukiyama, Ide, Ariyoshi, and Shirakawa~\cite{Tsukiyama77GeneratingAllTheMaxIndSets}
and later Johnson, Papadimitriou,  and Yannakakis~\cite{Johnson88MaxIndSet}
showed how to enumerate the maximal cliques of a graph $G = (V,E)$
with delay $O(|V| \,{\cdot}\, |E|)$.
We plug this into our framework
to obtain an algorithm to recognize conformal (i.e., 2-conformal) hypergraphs.
It improves over the previous $O(m^4 \nwspace n)$-time procedure based on \autoref{lem:BergeDuchet_subhypergraph}.

\begin{restatable}{theorem}{conformalitytest}
\label{thm:conformality_test}
	Conformal hypergraphs can be recognized in time $O(mn^3)$.
	\vspace*{.5em}
\end{restatable}

\noindent
\textbf{Outline.}
Next, we give a technical overview of our main contributions.
We fix the basic notation in \autoref{sec:prelims}.
\autoref{sec:delay_improvement} then describes the look ahead and its applications,
and \autoref{sec:equivalence} proves the equivalence between the decision and enumeration problems.

\section{Overview and Discussion}
\label{sec:overview}

Our first contribution is the look-ahead algorithm for the enumeration of minimal hitting sets
in hypergraphs of bounded transversal rank.
The same approach can be used to compute the transversal rank faster.
We then sketch the technical ideas behind the equivalence between the recognition of
hypergraphs with large transversal rank, those with large conformal degree,
and the enumeration of hypercliques and of hitting sets.
We also briefly discuss the possibility of further improvements
as well as open questions for future work.

\subsection{Looking Ahead for Higher-Order Extensions}
\label{subsec:overview_look-ahead}

Our enumeration algorithm
improves the one by Bläsius et al.~\cite{Blaesius22EfficientlyJCSS}.
It constructs a search tree pruned by a subroutine that decides,
for a pair of disjoint sets $X,Y \subseteq V$ of vertices,
whether $X$ can be extended to a minimal hitting set without using any elements from $Y$.
This extension subroutine runs in time $O(\Delta^{|X|} \nwspace mn)$ for any $Y$.
Their main argument for the $O(\Delta^{k^*} mn^2)$ delay bound is that,
on hypergraphs with transversal rank $k^*$,
the subroutine is only ever invoked with sets $X$ that satisfy $|X| \le k^*$.
This holds even though $k^*$ is a priori \emph{unknown} to the algorithm.
The article \cite{Blaesius22EfficientlyJCSS} also contains several conditional lower bounds
showing that the extension algorithm is almost optimal.
A method running in time 
$O(m^{|X|-\varepsilon}) \cdot \poly(n)$ for any constants $|X| \ge 2$ and $\varepsilon > 0$
would violate the Strong Exponential Time Hypothesis (SETH).
The, arguably much stronger,
Non-deterministic Strong Exponential Time Hypothesis (NSETH)
is a barrier to reduce the running time even to $O(m^{|X|+1-\varepsilon}) \cdot \poly(n)$
(see \autoref{cor:fine-grained} for a precise statement).
So there is little hope for improvement from that side.

Instead, we tackle the bottleneck of the tree search.
In the root the two sets $X$ and $Y$ are both empty.
The search checks in each node whether $(X,Y)$ is extendable.
If so, the algorithm branches on the decision whether to add the next vertex
in $V{\setminus}(X \cup Y)$ to the partial solution $X$ or the set of ``forbidden'' vertices $Y$.
The worst case occurs once a set $X$ is reached 
that has cardinality exactly equal to the transversal rank $k^*$.
If we were to know $k^*$ in advance,
we could stop the search one level before the bottleneck at $|X| = k^* -1$.
We would quickly find all extending minimal hitting sets
since they have at most one more vertex.
All other supersets of $X$ are irrelevant for the enumeration
as they are already too large to be minimal solutions.
The issue is that we cannot test for the right cardinality of $X$
without computing the transversal rank $k^*$.

To solve the bottleneck,
we instead integrate a ``look ahead'' in \emph{every} level of the search tree.
The main difficulty is to decide whether $X$ still has higher-order extensions,
that is, a minimal hitting set $T \supseteq X$ with $|T{\setminus}X| \ge 2$.
The search must then expand the current node $(X,Y)$ to find the remaining solutions.
We prove a structural characterization of the higher-order extensions,
which implies that deciding the existence of such a $T$
is at least as hard as the extension problem without size restrictions.
Notwithstanding, we give an algorithm running in the same $O(\Delta^{|X|} \nwspace mn)$ time.
In other words, our look ahead is for free.
This is the reason for our delay improvement
and also helps with the original problem 
of recognizing hypergraphs with transversal rank at least $k$.

One could hope that this approach can somehow be iterated.
Instead of considering only two more vertices,
one may want to check for minimal extension that have up to $\ell$ more elements 
and expect an enumeration delay of $\Delta^{k^*-\ell+1} \nwspace m n^{\ell + O(1)}$
and similar improvements for computing the transversal rank.
If feasible, the approach would bypass the lower bounds on the extension problem 
and, just by enhancing the tree search, enable a smooth trade-off 
between the dependency on $\Delta$ and $n$.
However, the iterated look ahead requires a complementing algorithm 
to decide whether $X$ is contained in a minimal hitting set of cardinality at least $|X|+\ell$.
For $\ell > 2$, this seems to be much harder than extension.

\subsection{Transversal Rank and Conformal Degree}
\label{subsec:overview_hypercliques}
	
It is conceivable that an entirely different approach lowers the running time
for recognizing hypergraphs with transversal rank at least $k$
all the way to $\poly(m) \cdot n^{k+O(1)}$.
We show that this is equivalent to several other new results
involving hypercliques instead of hitting sets.
Consider a hypergraph $\Hyp$ with vertex set $V$
and a minimal hitting set $S \subseteq V$ of at least $k$ vertices.
Minimality implies that,
for each subset $S' \subsetneq S$ with $|S'| = k-1$,
there exists an edge $E \in \Hyp$ such that $S' \cap E  = \emptyset$.
Thus, $S'$ is contained in the complement $\overline{E} = V{\setminus}E$.
Since $S$ hits every edge of $\Hyp$,
the set itself is \emph{not} contained in any complement edge.
This shows that the hypergraph $\overline{\Hyp} = \{V{\setminus}E \mid E \in \Hyp\}$
is not $(k{-}1)$-conformal.
In fact, \autoref{lem:BergeDuchet_subhypergraph} can be reformulated to
state that a transversal rank at least $k$ is equivalent
to the conformal degree of $\overline{\Hyp}$ being at least $k$
(see \autoref{lem:BergeDuchet_conformal}).

We extend this equivalence to several enumeration problems
by using ideas from the analysis of uniform hypergraphs.
Given a hypergraph $\Hyp$ with arbitrary edge sizes,
its $k$-\emph{section} $[\Hyp]_k$ is constructed
by including any set of $k$ vertices
that appear together in some edge of $\Hyp$ as an edge of $[\Hyp]_k$.
Clearly, the $k$-section is $k$-uniform.
(In the example above, each set $S'$ is an edge of the $(k{-}1)$-section of $\overline{\Hyp}$.)
We characterize the conformal degree of $\Hyp$
in terms of the maximal hypercliques of $[\Hyp]_k$.
The crucial observation is that we only need to know at most $|\Hyp|$ such hypercliques 
to compute the conformal degree, 
even if there are way more cliques than that.
This is where enumeration with incremental delay comes into play.
These algorithms allow for a non-trivial bound on the maximum time
between consecutive outputs,
and thus on the time until we see the $|\Hyp|$-th solution.

Recall that the cliques of a graph
are in one-to-one correspondence with the independent sets of the complement graph.
Those, in turn, correspond to the vertex covers (hitting sets) by taking the set complement.
The same argument transfers any delay improvement for minimal hitting sets in $k$-uniform hypergraphs
to the enumeration of maximal hypercliques.
Now to close the circle of reductions,
we show that a faster algorithm for computing the transversal rank
can indeed speed up the enumeration of minimal hitting sets in uniform hypergraphs.
Our result even holds for the larger class of hypergraphs with bounded rank,
we do not need all edges to have the same size.
Consider an algorithm that has already found several minimal hitting sets of $\Hyp$,
but has to decide whether there are still more solutions to be discovered.
The minimal hitting sets found so far form a hypergraph $\Gyp$ of their own.
Building on the work by Eiter and Gottlob~\cite{EiterGottlob95RelatedProblems},
we show how to decide whether $\Gyp$ already contains all solutions,
using information about the transversal rank of $\Gyp$.

As usual with fine-grained equivalences, 
they can be seen as additional motivation to find better algorithms
or as conditional evidence of the impossibility of this endeavor.
Currently, there is no algorithm known that computes the transversal rank,
equivalently, the conformal degree,
 in time $\poly(m) \cdot n^{k+O(1)}$.
We leave it as an open problem to obtain such a method or rule out its existence
under plausible complexity assumptions.
Arguably, a more modest research goal is
an algorithm running in time $m^{k+1-\varepsilon} \cdot n^{o(k)}$.

\section{Preliminaries}
\label{sec:prelims}

\textbf{Hypergraphs, rank, and complements.}
A \emph{hypergraph} is a finite set $V$ together with a collection $\Hyp \subseteq 2^V$ of subsets.
The elements of $V$ are the \emph{vertices}, $\Hyp$ contains the (\emph{hyper}-)\emph{edges}.
We usually identify the hypergraph with its edge set.
It is safe to assume that all hypergraphs in this work share the same vertex set.
We use $n = |V|$ to denote its size.
The hypergraphs may have different number of edges,
we use $m = |\Hyp|$ if this does not create ambiguities.
A vertex $v \in V$ has \emph{degree} $\deg(v) = | \{ E \in \Hyp \mid v \in E\}|$,
i.e., the number of edges containing $v$.
The \emph{maximum degree} of $\Hyp$ is $\Delta = \max_{v \in V} \deg(v)$.

The \emph{rank} $\rank(\Hyp) = \max_{E \in \Hyp} |E|$
is the maximum edge cardinality.
If the rank is bounded by some non-negative integer $r$,
we say that $\Hyp$ is an \emph{$r$-hypergraph}.
It is \emph{$r$-uniform} if all its edges have the same size $r$.
The $2$-uniform hypergraphs are the usual (undirected, simple) graphs.
A hypergraph $\Hyp$ is \emph{Sperner} if no two of its edges are contained in each other.
Note that uniform hypergraphs (of any rank) are always Sperner.

We need two different notions of a \emph{complement} of a hypergraph $(V,\Hyp)$.
We define $\overline{\Hyp} = \{ V{\setminus}E \mid E \in \Hyp \}$
to be the hypergraph containing the complement edges.
In the specific case where $\Hyp$ is $r$-uniform,
we let $\Hyp^\textsf{C} = \{ E \in \binom{V}{r} \mid E \notin \Hyp\}$
consist of all the $r$-subsets of vertices that are non-edges of $\Hyp$.
\vspace*{.5em}

\noindent
\textbf{Hitting sets and transversal rank.}
Let $(V,\Hyp)$ be a (not necessarily rank-bounded) hypergraph
and $X \subseteq V$ a set of vertices.
The subhypergraph $\Hyp(X)$ consists of all edges $E \in \Hyp$
that contain an element of $X$.
The remaining edges are in $\Hyp^{V{\setminus}X} = \Hyp{\setminus}\Hyp(X)$.
The notation alludes to the fact that $E \in \Hyp^{V{\setminus}X}$
is equivalent to $E \subseteq V{\setminus}X$.

We say a set $T \subseteq V$ of vertices
is a \emph{hitting set}, or \emph{transversal},
if it intersects every hyperedge,
meaning that $E \cap T \neq \emptyset$ holds for all $E \in \Hyp$
or, equivalently, $\Hyp(T) = \Hyp$.
The hitting set is (\emph{inclusion-wise}) \emph{minimal} 
if none of its proper subsets is also a hitting set.
Equivalently, a hitting set is minimal iff
for each $t \in T$,
there exists a \emph{private} edge $E \in \Hyp$ with $E \cap T = \{t\}$.
If some $t \in T$ does not have a private edge,
then the set $T$ is \emph{redundant}.

The minimal hitting sets of $\Hyp$
form another hypergraph on the same vertex set $V$,
the \emph{transversal hypergraph} $\Tr(\Hyp)$.
It is Sperner, 
regardless of whether $\Hyp$ has this property.
For two disjoint sets $X,Y \subseteq V$,
we use $\Tr(\Hyp)[X,Y]$ for the subhypergraph of all those minimal hitting sets $T \in \Tr(\Hyp)$
that satisfy $X \subseteq T \subseteq V{\setminus}Y$.

The \emph{transversal rank} of $\Hyp$ is $\rank(\Tr(\Hyp))$,
that is, the maximum cardinality of a minimal transversal.
We consider the corresponding search problem.
\vspace*{.25em}

\noindent
\TransRank
\vspace*{.25em}
\begin{description}
	\item [Input:] A hypergraph $(V,\Hyp)$ and a non-negative integer $k$.
	\vspace*{.25em}
	\item [Output:] Either produce $T \in \Tr(\Hyp)$ with $|T| \ge k$
			or correctly output that $\rank(\Tr(\Hyp)) < k$.
\end{description}
\vspace*{-.25em}

\noindent
The problem is trivial for $k \le 2$.
A hypergraph $\Hyp$ has transversal rank at least $1$ if and only if $\emptyset \notin \Hyp$.
It has transversal rank at least $2$ iff it contains two edges
that are incomparable with respect to set-inclusion.
\vspace*{.25em}

\noindent
\textbf{Hypercliques and conformality.}
Let $(V,\Hyp)$ be an $r$-uniform hypergraph.
A \emph{hyperclique} in $\Hyp$ is a set $C \subseteq V$ of at least $r$ vertices
such that every subset $S \subseteq C$ with $|S| = r$
is a hyperedge, $S \in \Hyp$.
A hyperclique is (\emph{inclusion-wise}) \emph{maximal} if all its proper supersets are no longer hypercliques.
The collection of all maximal hypercliques is denoted $\C(\Hyp)$.
Conversely, if $C$ does not contain any edge of $\Hyp$, it is an \emph{independent set}.

Let now $(V,\Hyp)$ be an arbitrary hypergraph and $k$ a positive integer.
$\Hyp$ is $k$-\emph{conformal} if for every set $S \subseteq V$ the following
two conditions are equivalent.
\vspace*{.25em}
\begin{enumerate}
	\item There exists a hyperedge $E \in \Hyp$ such that $S \subseteq E$.
	\vspace*{.25em}
	\item For each subset $S' \subseteq S$ with $|S'| \le k$ there exists an edge $E' \in \Hyp$ 
		with $S' \subseteq E'$.
\end{enumerate}
\vspace*{-.25em}

\noindent
If $\Hyp$ is $k$-conformal, then it is also $k'$-conformal for all $k' \ge k$.
The \emph{conformal degree} of $\Hyp$ is the smallest $k$ such that $\Hyp$ is $k$-conformal.
\vspace*{.25em}

\noindent
\ConfDeg
\vspace*{.25em}
\begin{description}
	\item [Input:] A hypergraph $(V,\Hyp)$ and a non-negative integer $k$.
	\vspace*{.25em}
	\item [Output:] Either correctly output that $\Hyp$ is $k$-conformal
		or produce a set $S \subseteq V$ such that\\[.25em]
		\hspace*{2.33em} every subset in $\binom{S}{\le k}$ is contained
		in some edge of $\Hyp$, but $S \nsubseteq E$ for all $E \in \Hyp$.
\end{description}

\noindent
\textbf{Parameterized Complexity and the (Strong) Exponential Time Hypothesis.}
We assume some familiarity with parameterized complexity
and only review the most important concepts.
For a more detailed exposition, we refer the reader to the textbooks~\cite{Cygan15ParamertizedAlgorithms,DowneyFellows13Parameterized,FlumGrohe06ParameterizedComplexityTheory}.
Consider pairs $(I,k) \in \{0,1\}^* \times \mathbb{N}_0$,
where $I$ is (a suitable encoding of) an input instance 
and the non-negative integer $k$ is the parameter.
A parameterized problem $\Pi$ is a subset of $\{0,1\}^* \times \mathbb{N}_0$.
The problem is \emph{fixed parameter tractable} (\FPT) if there exists a computable function $f$
such that the inclusion $(I,k) \in \Pi$ can be decided in time $f(k) \cdot \poly(|I|)$,
called \FPT-\emph{time}.
A parameterized reduction between parameterized problems $\Pi$ and $\mathrm{P}$
is a function computable in \FPT-time that maps instances $(I,k)$ for $\Pi$
to equivalent instances $(I',k')$ for $\mathrm{P}$
such that there exists a computable function $g$ with $k' \le g(k)$.

These reductions give rise to a system of complexity classes called the \W-\emph{hierarchy}.
The \textsc{Weighted Circuit Satisfiability} problem is,
given a Boolean circuit $C$ and a non-negative integer $k$,
to decide whether $C$ has a satisfying assignment of Hamming weight $k$.
The problem parameter is $k$.
A large gate in $C$ is one with fan-in strictly larger than $2$.
The \emph{weft} of $C$ is the maximum number of large gates on any input-output path through $C$.
For any positive integer $t$, the class $\W[t]$
is the collection of all parameterized problems reducable to
\textsc{Weighted Circuit Satisfiability} when restricted to circuits of constant depth and weft $t$.
All inclusions in the \W-hierarchy 
$\FPT \subseteq \W[1] \subseteq \W[2] \subseteq \W[3] \subseteq \dots$ 
are conjectured to be strict.
The class $\co\W[t]$ contains the complements of problems in $\W[t]$.
Beyond the \W-hierarchy lies the class \para-\NP
which contains all parameterized problems that can be solved by a \emph{nondeterministic} algorithm in \FPT-time.
We have $\FPT = \text{\para-\NP}$ if and only if $\textsf{P} = \NP$,
 see~\cite{FlumGrohe06ParameterizedComplexityTheory}.
A parameterized problem is complete for \para-\NP if it is \NP-complete
already for constant parameter values.
The class \para-\co\NP is the analog using co-nondeterministic algorithms and the class \co\NP.

A stronger hardness assumption than the $\W$-hierarchy being proper
is the \emph{Exponential Time Hypothesis} (ETH)~\cite{Impagliazzo01ETH1,Impagliazzo01ETH2}.
It states that the $3$-\textsc{Sat} problem,
the satisfiability problem for Boolean formulas on $n$ variables
in conjunctive normal form with clause length $3$, cannot be solved in time $2^{o(n)}$.
It is known that ETH implies $\W[1] \neq \FPT$~\cite{Chen06LowerBounds}.
In turn, ETH is implied by the \emph{Strong Exponential Time Hypothesis} (SETH)~\cite{Impagliazzo01ETH2}
positing that, for every constant $\varepsilon > 0$, there exists a clause length $\ell$
such that $\ell$-\textsc{Sat} cannot be solved in time $O(2^{(1-\varepsilon) n})$.
Finally, the \emph{Nondeterministic Strong Exponential Time Hypothesis} (NSETH)
is a common generalization of SETH and $\co\NP \neq \NP$~\cite{Carmosino16NSETH}.
It denotes the conjecture that, for every $\varepsilon > 0$, there exists an $\ell$
such that no \emph{co-nondeterministic algorithm without randomness} can decide $\ell$-\textsc{Sat}
in time $O(2^{(1-\varepsilon) n})$.
Note that, while it formalizes the current frontier in worst-case \textsc{Sat}-solving,
NSETH is a very strong assumption and thus not entirely plausible.
Its value stems from the fact that both proving and refuting NSETH would have
interesting consequences in complexity theory~\cite{Carmosino16NSETH}.

\section{A Look Ahead for the Transversal Rank}
\label{sec:delay_improvement}

In this section, we present an algorithm that,
given a non-empty hypergraph $\Hyp \neq \emptyset$ with (unknown) transversal rank $k^*$
and maximum degree $\Delta$, 
enumerates the minimal hitting sets of $\Hyp$ with delay 
$O(\Delta^{k^*-1} \nwspace mn^2)$
using $O(mn)$ space.
We first review the method by Bläsius et al.~\cite{Blaesius22EfficientlyJCSS},
including the analysis of Araújo et al.~\cite{Araujo23MaximumMinimalBlockingSetAlgorithmica}.

\subsection{The Enumeration Algorithm by Bläsius et al.}
\label{subsec:delay_Blaesius}

The algorithm in \cite{Blaesius22EfficientlyJCSS}
constructs a search tree which is pruned by repeated calls to an \emph{extension oracle},
that is, a subroutine that decides for set of vertices $X$
whether it is contained in a minimal hitting set $T \subseteq X$.
We call such a set $T$ a \emph{minimal extension} of $X$.
The correctness of the algorithm hinges on the following lemma.

\begin{lemma}[\cite{Blaesius22EfficientlyJCSS,Boros98Subimplicants}]
\label{lem:char_extendable}
	Let $(V,\Hyp)$ be a hypergraph and $X \subseteq V$ a set of vertices.
	The set $X$ has a minimal extension if and only if
	there exists a family of edges $\{E_x\}_{x \in X} \subseteq \Hyp$ such that
	\vspace*{.25em}
	\begin{enumerate}
		\item for every vertex $x \in X$, we have $E_x \cap X = \{x\}$;
		\vspace*{.25em}
		\item for every edge $E \in \Hyp$ contained in $\bigcup_{x \in X} E_x$,
			we have $E \cap X \neq \emptyset$.
	\end{enumerate}
\end{lemma}

\noindent
Recall that every vertex $x$ of a minimal hitting set needs a private edge only hit by $x$. 
The first condition of \autoref{lem:char_extendable} makes $E_x$ a \emph{candidate private edge}
for $x$ w.r.t.\ the partial solution $X$.
When $X$ is clear from context, we use $\Hyp_x = \{E \in \Hyp \mid E \cap X = \{x\}\}$
for the subhypergraph of all candidate private edges for $x$.
The second condition then ensures that the candidates $E_{x_1}, E_{x_2}, \dots$
for different elements of $X$ do not interfere with each other.

Besides the partial solution $X$, we also need to incorporate a set $Y$ of vertices
that have already been excluded from the search. 
With the notation $\Tr(\Hyp)[X,Y]$ for all minimal hitting sets of $\Hyp$
that contain $X$ but avoid $Y$,
we get the following.

\begin{corollary}
\label{cor:char_extendable}
	Let $(V,\Hyp)$ be a hypergraph and $X,Y \subseteq V$
	two disjoint set of vertices.
	There is a minimal hitting set $T \in \Tr(\Hyp)[X,Y]$ if and only if
	there exists a family of edges $\{E_x\}_{x \in X} \subseteq \Hyp$ such that
	\vspace*{.25em}
	\begin{enumerate}
		\item for every vertex $x \in X$, we have $(E_x{\setminus}Y) \cap X = \{x\}$;
		\vspace*{.25em}
		\item for every edge $E\in \Hyp$, if $(E{\setminus}Y) \cap X = \emptyset$,
		then $(E{\setminus}Y) \nsubseteq \bigcup_{x \in X} (E_x{\setminus}Y)$.
	\end{enumerate}
\end{corollary}

For a subhypergraph $\Gyp \subseteq \Hyp$,
we use $\Gyp^* = \{E{\setminus}Y \mid E \in \Gyp\}$
for the respective collection of edges that have $Y$ removed.
For example, $\Hyp^*_x = \{E{\setminus}Y \mid E \in \Hyp, E \cap X = \{x\}\}$
are the reduced candidate private edges for $x$,
and $(\Hyp^{V{\setminus}X})^* = \{E{\setminus}Y \mid E \in \Hyp, E \cap X = \emptyset \}$
the reduced edges not yet hit by $X$.
We sometimes also use $E^*$ for the edge $E{\setminus}Y$ itself.

\autoref{cor:char_extendable} implies the following solution to the extension problem for a given pair 
of disjoint sets $X$ and $Y$.
If one of the hypergraphs $\Hyp^*_x$ is empty, $X$ is redundant and thus has no minimal extension.
Otherwise, if $(\Hyp^{V{\setminus}X})^* = \emptyset$ (that is, all edges of $\Hyp^*$ are hit by $X$),
then $X $ is a minimal hitting set for the \emph{original} hypergraph $\Hyp$
of the non-reduced edges.
If those sanity checks do not already resolve the instance,
all combinations of edges
$(E^*_{x_1},E^*_{x_2},\dots, E^*_{x_{|X|}}) \in \prod_{x \in X} \Hyp^*_{x}$
in the Cartesian product are tested whether their union
$\bigcup_{x \in X} E^*_x$ completely contains some edge of $(\Hyp^{V{\setminus}X})^*$.
There is a combination for which this is \emph{not} the case if and only if 
$X$ is extendable to a minimal hitting set for $\Hyp$ without using $Y\!$.

Fix an arbitrary total order $\le$ on the vertex set $V$.
The algorithm constructs a tree in which
each node is labeled with a disjoint pair $(X,Y)$ 
such that $X \cup Y$ is a prefix of $V$ w.r.t.\ the order $\le$. 
The root is labeled $(\emptyset,\emptyset)$.
In any the node $(X,Y)$, it is tested whether $X$ is extendable without $Y$.
If $X$ has no minimal extension, the search in that branch halts and backtracks.
If the check determines that $X$ itself is already a minimal hitting set,
the solution is output and the branch pruned as well.
Otherwise, $X$ is indeed extendable without $Y$ and
the search branches on the decision where to include the next vertex 
$v_{\min} = \min_{\le} {V{\setminus}(X \cup Y)}$
by creating two child nodes $(X \cup \{v_{\min}\},Y)$ and $(X, Y\cup\{v_{\min}\})$.

The algorithm does not need to know the transversal rank.
By design, the cardinality of $X$ never exceeds $k^*$.
Any set with $|X|  = k^*$ is either a (maximum) minimal hitting set itself 
or already too large to be extendable.
The worst case occurs if a set with $k^*{-}1$ vertices is verified to be extendable,
but the next vertex $v_{\min}$ is such that 
$X \cup \{v_{\min}\}$ is neither redundant nor minimal.
Then, all $\prod_{x \in X \cup \{v_{\min}\}} |\Hyp^*_{x}| = O(\Delta^{k^*})$ possible combinations
have to be checked only to verify that there is always some edge
 $E^* \in (\Hyp^{V{\setminus}(X \cup \{v_{\min}\})})^*$
with $E^* \subseteq \bigcup_{x \in X \cup \{v_{\min}\}} E^*_x$.
Each check takes $O(|(\Hyp^{V{\setminus}(X \cup \{v_{\min}\})})^*| \cdot |V|) = O(mn)$ time.
The $O(\Delta^{k^*} mn^2)$ delay follows from $O(n)$ nodes sitting between two outputs.

\subsection{Higher-Order Extensions}
\label{subsec:delay_higher-order}

Our improvement stems from avoiding the bottleneck
by looking ahead and checking whether $X$ has a minimal extension with at least 2 more vertices.
Consider disjoint sets $X$ and $Y$ of which we already know
that there exists a minimal hitting set $T \in \Tr(\Hyp)[X,Y]$.
It trivial to check whether $T = X$ itself is already minimal,
we call this a \emph{0-extension}. 
It is also easy to find all \emph{1-extensions},
the minimal solutions that have exactly one more element than $X$.
However, the minimal extension $T$ may be (much) larger than $|X|+1$.
Those are the \emph{higher-order extensions}.
As it will turn out, deciding their existence is the hard core of the extension problem,
see \autoref{subsec:delay_complexity}.
The following lemma is crucial both for understanding the computational complexity
of higher-order extensions as well as for implementing the look-ahead idea 
without sacrificing running time.
We first discuss the case where $Y = \emptyset$.

\begin{lemma}
\label{lem:higher-order}
	Let $(V,\Hyp)$ be a hypergraph and $X \subseteq V$ a set of vertices.
	Define $S = \bigcap \, \Hyp^{V{\setminus}X}$.
	A minimal hitting set $T \in \Tr(\Hyp)$ with $T \supseteq X$ 
	satisfies $|T{\setminus}X| > 1$ if and only if $ T \subseteq V{\setminus}S$.
\end{lemma}

\begin{proof}
	The set $S = \bigcap \, \Hyp^{V{\setminus}X}$ contains those vertices
	that appear in all edges that are not yet hit by $X$.
	For any such $s \in S$, the set $X \cup \{s\}$ is a (not necessarily minimal) hitting set of $\Hyp$.
	Thus, no proper superset $T \supsetneq X \cup \{s\}$ is a \emph{minimal} hitting set.
	If $X$ has a higher-order extension $T$ at all, we thus have $s \notin T$.
	
	Conversely, every $T \in \Tr(\Hyp)$ with $X \subseteq T \subseteq V{\setminus}S$
	has at least two more elements than $X$.
	It cannot use a single vertex to hit all edges in $\Hyp^{V{\setminus}X}$
	since those are all in $S$.
\end{proof}

Regarding the general case that includes non-empty sets $Y$,
it is easy to verify that $T$ is a minimal extension of $X$ avoiding $Y$
if and only if it is an (arbitrarily sized) minimal extension in the hypergraph of reduced edges
$\Hyp^* = \{E{\setminus}Y \mid E \in \Hyp\}$.

\begin{corollary}
\label{cor:higher-order}
	Let $(V,\Hyp)$ be a hypergraph and $X,Y \subseteq V$
	two disjoint sets.
	Let $S = \bigcap \, (\Hyp^{V{\setminus}X})^*$ be the 
	intersection of the \emph{reduced} edges that are disjoint from $X$.
	A minimal hitting set $T \in \Tr(\Hyp)$ with $X \subseteq T \subseteq V{\setminus}Y$
	satisfies $|T{\setminus}X| > 1$
	if and only if $T \subseteq V{\setminus}(Y\,{\cup}\,S)$.
\end{corollary}

We now implement the look ahead.
Instead of merely checking whether $X$ is extendable without $Y$,
we devise an algorithm that outputs all 0- and 1-extensions of $X$ avoiding $Y$,
and then decides whether there are any higher-order extension left in $\Tr(\Hyp)[X,Y]$.
Moreover, if there are such large extensions,
the algorithm computes a set $Y^{+} \supseteq Y$ with the guarantee that all remaining extensions
also avoid $Y^{+}$.
This obviates the need to evaluate nodes $(X,Z)$ with $Y \subseteq Z \subseteq Y^+$
and thus allows us to short-cut the search.
We further include several sanity checks that allow for premature pruning on suitable inputs.
A full overview of the procedure is given in \autoref{alg:ExtHS}.

\begin{algorithm2e}[p]
\setstretch{1.11}
\vspace*{.5em}
	\Input {\hspace{.53em} Disjoint sets $X = \{x_1, \dots, x_{|X|}\}$ and $Y \subseteq V$.}
	\Output{All minimal extensions $T \in \Tr(\Hyp)[X,Y]$ with $|T{\setminus}X| \le 1$.}
	\Ret{\hspace{.07em}\textsc{continue with} $Y^{+}$ if there exists a $T'\in \Tr(\Hyp) [X,Y]$ 
		with $|T'{\setminus}X| > 1$,\\ 
		\hspace*{4.12em} where $Y \subseteq Y^{+} \subseteq V{\setminus}X$
		is such that all such $T'$ are in $\Tr(\Hyp)[X,Y^{+}]$.\\
		\hspace*{4.1em} \halt if there is no such $T'$.}
	\vspace*{.5em}
	\nonl\Proc{\Ext{$X, Y$}}{:\\
	
	\vspace*{.5em}
	\tcc{preprocessing empty sets $X$}
	\If {$X = \emptyset$}
	{\label{line:check_empty}
		initialize $\Hyp^* = \emptyset$\;
		\ForEach {$E \in \Hyp$}{%
			\lIf{$E \subseteq Y$}{\Return \halt}\label{line:unhit}
			\lIf{$E \cap X = \emptyset$}{add $E{\setminus}Y$ to $\Hyp^*$}
		}
		$S \gets \bigcap \Hyp^*$\;
		\lForEach{$s \in S$}{\out $\{s\}$}
		\lIf{$|\Hyp^*|  = 1$}{\Return \halt}
		\Return \textsc{continue with} $Y \cup S$\;
	} \label{line:end_empty}
	
	\vspace*{.5em}
	\tcc{preprocessing non-empty sets $X$}
	initialize hypergraph $(\Hyp^{V{\setminus}X})^* = \emptyset$\; \label{line:non-empty}
	\lForEach {$x \in X$} {initialize hypergraph $\Hyp^*_x = \emptyset$} 
	\ForEach {$E \in \Hyp$}
	{
		\lIf{$E \subseteq Y$}{\Return \halt} \label{line:not_contained}
		\lIf{$E \cap X = \emptyset$}{add $E{\setminus}Y$ to $(\Hyp^{V{\setminus}X})^*$}
		\lIf {$E \cap X = \{x \}$} {add $E{\setminus}Y$ to $\Hyp^*_x$}	\label{line:check_for_witnesses}
	}
	\lIf {$\exists \nwspace x \in X \colon \Hyp^*_x = \emptyset$}{\Return \halt}	\label{line:everyone_has_a_witness}
	\If {$(\Hyp^{V{\setminus}X})^* = \emptyset$ \label{line:previous_check}}{%
		\out $X$\; 
		\Return \halt;
	}	\label{line:end_of_preprocessing}
	
	\vspace*{.5em}
	\tcc{look ahead}
	$S \gets \bigcap \nwspace (\Hyp^{V{\setminus}X})^*$\;			\label{line:start_of_lookahead}
	$U \gets \bigcup_{x \in X} \nwspace \bigcap \Hyp_x^*$\; \label{line:retain_candicacy}
	\lForEach{$s \in S{\setminus}U$}{%
		\out $ X \cup \{s\}$						\label{line:end_of_lookahead}
	}	
	
	\vspace*{.5em}
	\tcc{higher-order extensions}
	\lIf{$|(\Hyp^{V{\setminus}X})^*| = 1$}{\Return \halt}  \label{line:higher_order_sanity}
	\ForEach {$(E^*_{x_1}, \ldots, E^*_{x_{|X|}}) \in \Hyp^*_{x_1} \times \dots \times \Hyp^*_{x_{|X|}}$	\label{line:brute_force}}
	{
		$W \gets S \cup \bigcup_{i = 1}^{|X|} E^*_{x_i}$\;
		\lIf{$\forall \nwspace E^* \in (\Hyp^{V{\setminus}X})^* \colon E^* \nsubseteq W$}{%
			\Return \textsc{continue with} $Y \cup S \cup U$}	\label{line:combine_witnesses}
	} 	
	\Return \halt\;	\label{line:end}
	}
	\vspace*{.5em}
\caption{Computing the 0- and 1-extensions of $X$ avoiding $Y$ 
	in a hypergraph $(V,\Hyp)$.}
\label[alg]{alg:ExtHS}
\end{algorithm2e}

Until \autoref{line:end_empty}, the procedure handles the case $X = \emptyset$.
This part consists of simplifications of all the major steps that are also used for non-empty $X$.
We defer its description until we have discussed the general case in more detail.
Assume $X \neq \emptyset$ for now.
In lines~\ref{line:non-empty} to \ref{line:check_for_witnesses},
data structures are preprocessed that hold
the subhypergraphs $\Hyp^*_x$ for all $x \in X$ as well as $(\Hyp^{V{\setminus}X})^*$.
The the former hold the candidate private edges while 
the latter contains those $E{\setminus}Y$ for which $E \in \Hyp$ is disjoint from $X$.
Edges $E$ with $|E \cap X| \ge 2$ can be discarded.
If the check for $E \subseteq Y$ in \autoref{line:not_contained} is triggered,
then $E{\setminus}Y = \emptyset$ would be contained in $(\Hyp^{V{\setminus}X})^*$.
This certifies that $X$ cannot be extended without $Y$.
(See the second item of \autoref{cor:char_extendable}.)
The execution can thus be aborted outright. 
After the data structures are filled,
the procedure verifies in \autoref{line:everyone_has_a_witness} 
that every $x$ indeed has a candidate private edge,
otherwise $X$ is redundant.
If $(\Hyp^{V{\setminus}X})^* = \emptyset$, there are no unhit edges.
The previous checks for private edges ensure that 
$X$ must then be a \emph{minimal} hitting set.
Consequently, $X$ is output as its own 0-extension.

The following part from \autoref{line:start_of_lookahead} 
to \ref{line:end_of_lookahead} is the look ahead.
It decides whether adding a single vertex to $X$ is enough to make it a minimal solution.
By \autoref{lem:higher-order}, if there exists such a vertex $s$,
it must be contained in $S = \bigcap \,(\Hyp^{V{\setminus}X})^*$.
However, this condition is not sufficient,
the hitting set $X \cup \{s\}$ may be non-minimal.
The set is indeed minimal iff every $x \in X$ has a (candidate) private edge
w.r.t.\ $X$
that does not contain $s$.
The same edges are then is also private w.r.t.\ $X \cup \{s\}$.
This is tested in lines~\ref{line:retain_candicacy} to \ref{line:end_of_lookahead}.
Some vertex $x \in X$ having a private edge that does not contain $s$
is equivalent to $s \notin \bigcap \Hyp_x^*$.
Therefore, $s \notin U = \bigcup_{x \in X} \nwspace \bigcap \Hyp_x^*$
verifies this for all $x$ simultaneously.
In principle, the new element $s$ also needs a private edge.
However, \emph{any} edge in $(\Hyp^{V{\setminus}X})^*$ can serve this purpose since they all contain $s$
and are disjoint from $X$.
It has already been checked in \autoref{line:previous_check}
that $(\Hyp^{V{\setminus}X})^*$ is non-empty.
In summary, $X\cup \{s\}$ is certified as a 1-extension for every $s \in S{\setminus}U$.

The remaining lines decide whether $X$ has any higher-order extensions,
starting with a sanity check in \autoref{line:higher_order_sanity}.
Consider the case in which there is only a single edge left that is not yet hit by $X$,
that is, $|(\Hyp^{V{\setminus}X})^*| = 1$.
We thus have $(\Hyp^{V{\setminus}X})^* = \{S\}$ and $\Tr(\Hyp)[X, Y \,{\cup}\,S]$ is empty. 
In other words, $X$ has no higher-order extensions.
For the remainder, suppose that there are more than one edge in $(\Hyp^{V{\setminus}X})^*$
that we need to hit.
By \autoref{cor:higher-order}, $X$ has a higher-order extension avoiding $Y$ if and only if
$X$ can be extended without using any vertex of $Y \cup S$.
\autoref{cor:char_extendable} in turn states that this is equivalent
to the existence of a combination of edges 
$E_{x_1} \in \Hyp_{x_1}, \dots, E_{x_{|X|}} \in \Hyp_{x_{|X|}}$ such that
\begin{equation}
\label{eq:avoiding_S}
	E{\setminus}(Y \cup S) \nsubseteq \bigcup_{i = 1}^{|X|} E_{x_i}{\setminus}(Y \cup S)
\end{equation}
holds for every $E \in \Hyp^{V{\setminus}X}$.
We avoid computing the sets $E{\setminus}(Y \,{\cup}\, S)$ or $E_{x_i}{\setminus}(Y \,{\cup}\, S)$
explicitly and instead reuse
$E^* = E{\setminus}Y \in (\Hyp^{V{\setminus}X})^*$ and 
$E_{x_i}^* = E_{x_i}{\setminus}Y \in \Hyp_{x_i}^*$
that were already prepared during the preprocessing.
Observe that (\ref{eq:avoiding_S}) is equivalent to 
$E{\setminus}Y \nsubseteq S \cup  \bigcup_{i = 1}^{|X|} E_{x_i}{\setminus}(Y \cup S)$
and therefore to
$E^* \nsubseteq S \cup  \bigcup_{i = 1}^{|X|} E_{x_i}^*$.
If $X$ indeed has higher-order extensions avoiding $Y$,
then the previous computations show that they do not use vertices from $S$
(otherwise they would be a 1-extensions)
and also not from $U$ (any set $X \cup \{u\}$ with $u \in U$ is redundant).

We now return to the case $X = \emptyset$ (lines~\ref{line:check_empty} to \ref{line:end_empty}).
Again, a data structure for $(\Hyp^{V{\setminus}X})^*$ is constructed and
the algorithm computes the set $S$.
The hypergraphs $\Hyp^*_x$ are obviously not needed.
In this special case, the containment $s \in S$ is \emph{both} necessary
and sufficient for $X \cup \{s\} = \{s\}$ to be a 1-extension.
The unique proper subset $\emptyset \subsetneq \{s\}$ is not a hitting set
by our assumption that $\Hyp \neq \emptyset$ is non-empty.

The brute-force loop (lines~\ref{line:brute_force} to \ref{line:combine_witnesses})
dominates the running time.
Our procedure thus runs in the same asymptotic time $O(\Delta^{|X|} \nwspace mn)$
as the one by Bläsius et al.~\cite{Blaesius22EfficientlyJCSS}
(see also the analysis of Araújo et al.~\cite{Araujo23MaximumMinimalBlockingSetAlgorithmica}).
Since all the data structures $\Hyp^*_{x}$ and $(\Hyp^{V{\setminus}X})^*$
hold disjoint subhypergraphs of $\Hyp^*$,
the total space requirement is $O(mn)$.
Computing all 1-extensions of $X$ as well as the continuation set $Y^+$ 
does not incur a penalty in the running time.

\begin{algorithm2e}[t]
\setstretch{1.1}
\vspace*{.25em}
	\Input{\hspace{.85em}Disjoint sets $X,Y \subseteq V$.}
	\Output{The minimal hitting sets $T \in \Tr(\Hyp)[X,Y]$.}
	\vspace*{.5em}
	\Proc{\Enum{$X, Y$}}{:\\
		\emph{isExtendable} $\gets\ $\Ext{$X, Y$}\; \label{line:call_extendable}
		\lIf{ isExtendable $==$ \emph{\textsc{halt}}}{\Return}
		\If{ isExtendable $==$ \emph{\textsc{continue with}} $Y^{+}$}{
			$v_{\min} \gets \min_{\le} V{\setminus}(X \cup Y^{+})$\;
			\Enum{$X \cup \{v_{\min}\}, Y^{+}$}\; \label{line:recursive_call_1}
			\Enum{$X, Y^{+} \cup \{v_{\min}\}$}\; \label{line:recursive_call_2}
		}
	}
\caption{Recursive algorithm for the transversal hypergraph problem of a hypergraph $\Hyp$ 
	on an ordered vertex set $(V,\le)$.
	The initial call is \mbox{\texttt{enumerate(}$\emptyset$, $\emptyset$\texttt{)}}.}
	\label[alg]{alg:enumerate_recursively}
\end{algorithm2e}

\subsection{Delay Bound}
\label{subsec:delay_proof}

\autoref{alg:enumerate_recursively} combines the \Ext procedure from 
\autoref{alg:ExtHS} with the tree search.
We now show its correctness and the delay improvement.
The technical parts of the proof are encapsulated in the following lemma.

\pagebreak

\begin{lemma}
\label{lem:technical_delay}
	Let $X,Y \subseteq V$ be two disjoint sets of vertices.
	\vspace*{.25em}
	\begin{enumerate}
		\item A call to \Enum{$X, Y$} enumerates exactly the minimal hitting sets in $\Tr(\Hyp)[X,Y]$.
		\vspace*{.25em}
		\item During the execution of \autoref{alg:enumerate_recursively},
			every call to \Enum{$X, Y$} satisfies $|X| \le k^* - 1$.
	\end{enumerate}
\end{lemma}

\begin{proof}
	We prove the first clause under the assumption that a relaxation of the second one indeed holds,
	namely, that $|X| \le k^*$.
	This enables a proof by induction over $k^* - |X|$.
	If $X$ is precisely as large as transveral rank, then either $X$ is either a minimal hitting set
	(its own 0-extension)
	or does not have a minimal extension.
	The call to \Ext{$X, Y$} in \autoref{line:call_extendable} therefore returns \textsc{halt},
	possibly after outputting the solution $X$.
	Procedure \Enum{$X, Y$} then returns without any further actions.
	
	Now assume $|X| < k^*$.
	If all solutions in $\Tr(\Hyp)[X,Y]$ are 0- or 1-extensions of $X$,
	the claim follows as before since \Ext{$X, Y$} outputs those and then returns \textsc{halt}.
	Otherwise, it returns \textsc{continue with} $Y^+$ for some set 
	$Y \subseteq Y^+ \subseteq V{\setminus}X$ such that all higher-order extensions in $\Tr(\Hyp)[X,Y]$ 
	are in fact in $\Tr(\Hyp)[X,Y^+]$.
	In particular, there exists at least one such extension and 
	therefore $V{\setminus}(X \cup Y^+)$ is non-empty.
	Consider the next vertex $v_{\min} = \min_{\le} V{\setminus}(X \cup Y^+)$.
	Inductively, every solution that contains $v_{\min}$ is enumerated by the call
	to \Enum{$X \cup \{v_{\min}\}, Y^+$} in \autoref{line:recursive_call_1},
	the remaining solutions are found by \Enum{$X, Y^+ \cup \{v_{\min}\}$}
	in \autoref{line:recursive_call_2}.
	
	The second clause states that $|X| \le k^* - 1$ holds throughout.
	To reach a contradition, assume that at some point \Enum{$X, Y$} is called for some
	set $X$ with at least $k^*$ elements.
	Since $\Hyp$ is non-empty, $k^*$ is positive and thus $(X,Y) \neq (\emptyset,\emptyset)$
	is not the root of the search tree.
	Let $(A,B)$ be the parent node of $(X,Y)$.
	In particular,
	during the execution of \Enum{$A,B$},
	the call to \Ext{$A, B$} (\autoref{line:call_extendable})
	returned \textsc{continue with} $B^+$,
	meaning, we have $B^+ = Y$.
	Therefore, there exists a minimal hitting set in $T \in \Tr(\Hyp)[A,Y]$
	that has at least two more vertices than $A$.
	The sets $A$ and $X$ are either equal
	or $X$ has exactly one additional vertex, namely, 
	$v_{\min}  = \min_{\le} V{\setminus}(A \cup Y)$.
	The cardinality of that minimal extension must therefore be 
	\begin{equation*}
		|T| \ge |A|+2 \ge |X| +1 \ge k^* + 1.
	\end{equation*}
	This is a contradiction to $k^*$ being the maximum cardinality
	of any minimal hitting.
\end{proof}


We now have the tools ready for the delay bound.

\begin{lemma}
\label{lem:delay_bound}
	Let $(V,\Hyp)$ be a non-empty hypergraph with $n$ vertices, $m$ edges,
	maximum degree $\Delta$, and transversal rank $k^*$. 
	\autoref{alg:enumerate_recursively} on input $\Hyp$ outputs every edge of $\Tr(\Hyp)$ exactly once
	with delay $O(\Delta^{k^*-1} \nwspace mn^2)$. The execution takes $O(mn)$ space.
\end{lemma}

\begin{proof}
	The initial call of \autoref{alg:enumerate_recursively} is \Ext{$\emptyset, \emptyset$}.
	The first clause of \autoref{lem:technical_delay} thus implies that 
	$\Tr(\Hyp) = \Tr(\Hyp)[\emptyset, \emptyset]$ is enumerated entirely.
	The $O(mn)$ space requirement follows from the
	the extension subroutine using this much space and
	the tree search merely maintaining the current pair $(X,Y)$.
	
	Regarding the delay, observe that between a parent and child node of the search tree
	there is always at least one new node added to either $X$ or $Y$.
	The depth of the tree is thus at most $n$.
	The tree is explored in preorder and thus at most $O(n)$ calls to \Ext{$X, Y$}
	are executed before the first output, between any two consecutive outputs,
	or after the last output until termination.

	In the worst case, the next output occurs for a set $X$ with $|X| = k^*-1$.
	Any of the $O(n)$ calls in-between thus take time $O(\Delta^{k^*-1} \nwspace mn)$,
	which amounts to a delay of $O(\Delta^{k^*-1} \nwspace mn^2)$.
	Note that this also holds for the time between the last output and termination
	since, in the worst case, 
	for each node on the path from the source $(\emptyset,\emptyset)$ to the last leaf,
	the algorithm needs to verify that the \emph{other} child node indeed does not lead to
	any solutions.
\end{proof}


\subsection{Computational Complexity}
\label{subsec:delay_complexity}

The previous section raised two computational problems.
Most prominently, we needed an algorithm to decide whether $X$ has a higher-order extension.
We refer to this as ``Problem~(a)''.
The other one is a bit more subtle.
Consider the recursive calls to \Enum in lines~\ref{line:recursive_call_1} 
and \ref{line:recursive_call_2} of \autoref{alg:enumerate_recursively}.
The second argument $Y^+ = Y \cup S \cup U$ states that, 
if $X$ has a higher-order extension avoiding $Y$ at all,
none of them contains a vertex of $S \cup U$.
In the same fashion, it would speed up the tree search if we could identify a set $X^+ = X \cup W$
such that, if $X$ is extendable to a minimal hitting set at all,
then $W$ is contained in each such extension (``Problem~(b)'').
We now settle the computational complexity of both problems.
It turns out that they are, respectively, complete for the classes \NP and $\W[3]$
as well as \co\NP and $\co\W[3]$.
All our reductions use the parameterized extension problem
for minimal hitting sets.
\vspace*{.5em}

\noindent
\ExtHS
\vspace*{.25em}
\begin{description}
	\item [Input:] A hypergraph $(V,\Hyp)$ and a set $X \subseteq V$ of vertices.
	\vspace*{.25em}
	\item [Parameter:] The cardinality $|X|$.
	\vspace*{.25em}
	\item [Decision:] Is $X$ contained in a minimal hitting set for $\Hyp$.
\end{description}

\noindent
Note \ExtHS does not require that the sought minimal hitting set extending $X$ has a certain size.
The problem is \NP-complete and $\W[3]$-complete,
see~\cite{Blaesius22EfficientlyJCSS,Boros98Subimplicants}.

\begin{theorem}
\label{thm:W3_higher-order}
	Deciding whether a set $X$ has a higher-order minimal extension (Problem~(a))
	is \emph{\NP}-complete, and \emph{{\W}[3]}-complete when parameterized by the cardinality $|X|$.
\end{theorem}

\begin{proof}
	As mentioned, we reduce from and to the \ExtHS.
	Its complexity is unaffected by the input including an additional set $S$ 
	that the sought minimal solution has to avoid.
	We can simply remove $S$ from every edge of the input hypergraph,
	see~\cite{Blaesius22EfficientlyJCSS}.	
	We show that the higher-order problem is equivalent to the (unrestricted) extension problem
	under polynomial-time reductions that preserve the parameter.

	Since $S = \bigcap \, \Hyp^{V{\setminus}X}$ is computable from $X$ and $\Hyp$
	in polynomial time,
	\autoref{lem:higher-order} already spells out the reduction
	that proves containment in $\NP$ and $\W[3]$.
	For the hardness, let $(V,\Hyp)$ be a hypergraph 
	and $X \subseteq V$ the set for which we have to decide the existence
	of a minimal extension without any cardinality restrictions.
	We check whether $X$ is itself a minimal hitting set 
	or can be turned to one by adding a single vertex from $V$.
	If so, the reduction produces a trivial yes-instance for the higher-order variant.
	Otherwise, we decide the existence of an higher-order extension
	for the original instance $(V,\Hyp,X)$.
\end{proof}

\begin{theorem}
\label{thm:W3_include_U}
	Deciding for disjoint sets $X, W$ whether every minimal extension of $X$ includes $W$
	(Problem~(b))
	is \emph{\co\NP}-complete, and \emph{{\co\W}[3]}-complete
	when parameterized by $|X|$.
	The problem is \emph{\para-\co\NP}-complete when paramterized by $|W|$ instead. 
	The complexity is unaffected by the promise that $X$ indeed has a minimal extension.
\end{theorem}

\begin{proof}
	We establish the theorem by reducing the yes-instances of the ``include-$W$'' extension problem
	to \emph{no}-instances of \ExtHS, and vice versa.
	One reduction is immediate from the observation that every minimal extension of $X$ contains $W$
	if and only if $X$ \emph{cannot} be extended without $W$.
	As mentioned in the proof of \autoref{thm:W3_higher-order},
	augmenting the input of \ExtHS with a set the extension needs to avoid does not change anything.
	
	To show hardness,
	let $(V,\Hyp,X)$ be an instance of \ExtHS.
	If $X$ is a minimal hitting set, we output a trivial 
	no-instance of Problem~(b).
	If this does not resolve the instance, we proceed as follows.
	Recall that $\Hyp(X)$ is the subhypergraph of those edges that intersect $X$.
	Since $X$ is not a hitting set, there is at least one unhit edge in $\Hyp^{V{\setminus}X}$.
	Let now $W = \{w\}$ be a singleton set whose unique element $w$ is not in $V\!$.
	We define a hypergraph $\Gyp$ on the vertex set $V \cup W$ by setting
	\begin{equation*}
		\Gyp = \Hyp(X) \cup \{E \cup W \mid E \in \Hyp^{V{\setminus}X}\}.
	\end{equation*}

	Every $x \in X$ has a private hyperedge in $\Hyp(X) \subseteq \Gyp$,
	so $X \cup W$ is a minimal hitting set for $\Gyp$.
	In particular, the promise that $X$ is extendable in $\Gyp$ holds.
	The set $X$ is \emph{not} extendable in the original hypergraph $\Hyp$ if and only if
	every minimal extensions in $\Gyp$ must include $W$.
	The fact that the problem is already hard for $|W| = 1$ gives the \para-\co\NP-completeness.
\end{proof}

Observe that the two hardness reductions can be computed in time $O(mn^2)$,
do not increase the number of edges $m$ or parameter $|X|$,
and add at most one new vertex.
Hence, we inherit all fine-grained lower bounds from the original extension problem
\cite[Theorem~4]{Blaesius22EfficientlyJCSS}.
In particular, our $O(\Delta^{|X|} \nwspace mn)$ running time for the higher-order extensions is tight,
up to possibly a factor $m$ if randomization is introduced.

\begin{corollary}
\label{cor:fine-grained}
	There exists \emph{no} algorithm to solve any of the Problems~(a) or (b)
	\vspace*{.25em}
	\begin{enumerate}
		\item in time $f(|X|) \cdot (m{+}n)^{o(|X|)}$ for any computable function $f$,
			unless $\emph{\W}[2] = \emph{\FPT}$;
		\vspace*{.25em}
		\item in time $m^{|X|-\varepsilon} \cdot \poly(n)$ 
			for any constant $|X| \ge 2$ and $\varepsilon > 0$,\\[.25em]
			unless the Strong Exponential Time Hypothesis fails.
	\end{enumerate}
	\vspace*{.25em}
	Moreover, no \emph{deterministic} reduction can prove SETH-hardness 
	of (a) or (b) with time bound $m^{|X|+1-o(1)} \cdot \poly(n)$,
	unless the Nondeterministic Strong Exponential Time Hypothesis fails.
\end{corollary}

\subsection{Computing the Transversal Rank}
\label{subsec:delay_TR}

We initially set out to find a way to compute the transversal rank $k^* = \rank(\Tr(\Hyp))$ faster.
We show that the same look-ahead idea can help here as well.

\transversalrankalg*

\begin{proof}
	We first discuss the trivial cases of the \textsc{Transversal Rank} problem
	with parameters $k = 0$ or $k = 1$.
	Any non-empty hypergraph with $\emptyset \notin \Hyp$ has positive transversal rank.
	In more detail, if $\emptyset \in \Hyp$ is an edge, then $\Tr(\Hyp) = \emptyset$ is empty
	and the transversal rank is undefined.
	If $\Hyp = \emptyset$, then $\Tr(\Hyp) = \{\emptyset\}$ and the transversal rank is $0$.
	Otherwise, $k^* \ge 1$ holds.
	Note that for $k = 0$ and $k = 1$, any minimal hitting set witnesses $\rank(\Tr(\Hyp)) \ge k$.
	Computing it takes $O(mn)$ time.

	In the remainder, we assume $k \ge 2$.
	The inequality $k^* \ge k$ is thus equivalent to the existence of a set of $k-2$ vertices
	that has a \emph{higher-order} minimal extension.
	(This also includes the empty set in case $k = 2$.)
	Consider a set $X \subseteq V$ of cardinality $|X| = k-2$
	and define $S = \bigcap \Hyp^{V{\setminus}X}$.
	By \autoref{lem:higher-order}, $X$ is contained in a minimal hitting set of size 
	at least $|X|+2$
	if and only if $\Tr(\Hyp)[X,S]$ is non-empty.
	Using the same arguments that already proved the correctness of 
	(the second half of) \autoref{alg:ExtHS},
	$\Tr(\Hyp)[X,S] \neq \emptyset$ can be decided in time 
	$O(\Delta^{|X|} mn) = O(\Delta^{k-2} \nwspace mn)$.
	Applying this procedure to all \mbox{$(k{-}2)$-sets} of vertices 
	then decides \textup{\textsc{Transversal Rank}} in total time
	$O(\Delta^{k-2} \nwspace mn^{k-1})$.
	
	If the transversal rank is indeed above the given threshold,
	we find a witness as follows.
	Let $X = \{x_1, \dots, x_{k-2}\}$ be the first set 
	the algorithm verifies to be contained in a sufficiently large hitting set.
	The brute-force loop (starting in \autoref{line:brute_force} of \autoref{alg:ExtHS})
	thus computed a suitable selection of hyperedges
	$E_{x_1} \,{\in}\, \Hyp_{x_1}, \dots, E_{x_{k-2}} \,{\in}\, \Hyp_{x_{k-2}}$
	such that, for every edge $E \in \Hyp^{V{\setminus}X}$,
	we have $E \nsubseteq S \cup \bigcup_{i=1}^{k-2} E_{x_i}$.
	Therefore, the set
	\begin{equation*}
		M = X \cup V \nwspace {\setminus} \!\left(S \cup \bigcup_{i=1}^{k-2} E_{x_i}\right)
	\end{equation*}
	is a (not necessarily minimal) hitting set of $\Hyp$.
	We claim that \emph{any} minimal hitting set $T \subseteq M$
	is an extension of $X$ with at least $|X|+2 = k$ elements.
	It can be computed from $\Hyp$ and $M$ in time $O(mn)$, see \autoref{app:intro_proofs}.	
	To verify the claim, first observe that any such $T$ must contain all of $X$.
	If some $x_i$ were missing, the edge $E_{x_i}$ would be unhit.
	Moreover, since $T$ is disjoint from $S$, it must have at least 2 more elements than
	$X$ by \autoref{lem:higher-order}.
\end{proof}

\section{Conformal Degree and Hyperclique Enumeration}
\label{sec:equivalence}

We have already seen a connection between computing the transversal rank
and enumerating minimal hitting sets.
We now show a quantitative equivalence to computing the conformal degree
and enumerating maximal hypercliques.

\mainequivalence*

Some of the implications in \autoref{thm:main_equivalence} have been known before
or can easily be derived from prior work.
\TransRank and \ConfDeg are sister problems in the following sense.
Recall that $\overline{\Hyp} = \{V{\setminus}E \mid E \in \Hyp\}$ 
is the hypergraph of the complement edges.
The following is a reformulation of \autoref{lem:BergeDuchet_subhypergraph}
in terms of conformal hypergraphs.

\begin{lemma}[Berge and Duchet~\cite{BergeDuchet75GilmoresTheorem};
see Corollary~1, p.~58 of~\cite{Berge89Hypergraphs}]
\label{lem:BergeDuchet_conformal}
	The transversal rank of $\Hyp$ is at least $k$
	if and only if $\overline{\Hyp}$ has conformal degree at least $k$.
\end{lemma}

Note that having a conformal degree of at least $k$ is the same as saying that
$\overline{\Hyp}$ is \emph{not} $(k{-}1)$-conformal.
The complement $\overline{\Hyp}$ can be obtained from $\Hyp$ in time $O(mn)$.
We can thus solve \TransRank with parameter $k$ in time $\poly(m) \cdot n^{k+O(1)}$
if and only if \ConfDeg with parameter $k-1$ is solvable within the same time bound. 
It is also readily checked that the required witnesses for $\rank(\Tr(\Hyp)) \ge k$
and the counterexamples for the $(k{-}1)$-conformality of $\overline{\Hyp}$
are the same.
Indeed, let $S$ be such that all its subset of size at most $k-1$ 
are contained in an edge of $\overline{\Hyp}$.
This is equivalent to none of those subsets being a hitting set for $\Hyp$.
Conversely, $S$ itself not being contained in any edge of $\overline{\Hyp}$
means that $S$ is a hitting set.
Any minimal hitting set $T \in \Tr(\Hyp)$ with $T \subseteq S$ 
must thus have size at least $k$.
One such subset $T$ can be obtained from $S$ in time $O(m |S|)$ 
(see \autoref{app:intro_proofs}).

Now assume we can enumerate the minimal hitting sets of any $r$-hypergraph
in incremental-polynomial time so that the $i$-th delay is bounded by $\poly(i) \cdot n^{r+O(1)}$.
Clearly, this also includes $r$-\emph{uniform} hypergraphs.
A straightforward generalization from graphs ($2$-uniform hypergraphs) shows 
how this can be used to enumerate maximal hypercliques.
Let $(V,\Hyp)$ be a uniform hypergraph with rank $r$.
The complement $\Hyp^\textsf{C} = \{ E \in \binom{V}{r} \mid E \notin \Hyp\}$
contains all the $r$-subsets of $V$ that are not edges of $\Hyp$.
Let $C \subseteq V$ be a maximal hyperclique in $\Hyp$,
then $C$ is a maximimal independent set in $\Hyp^\textsf{C}$, 
meaning that $C$ does not contain any of the complement edges.
This is equivalent to $V{\setminus}C$ being a minimal hitting set of $\Hyp^\textsf{C}$.
Conversely, for any $T \in \Tr(\Hyp^\textsf{C})$, $V{\setminus}T$ is a maximal clique of $\Hyp$.
The complement $\Hyp^\textsf{C}$ can be computed in time $O(mn^r)$,
hence $\C(\Hyp)$ is enumerable with delay $\poly(i) \cdot n^{r+O(1)}$
iff $\Tr(\Hyp^\textsf{C})$ is.

\subsection{Faster Decision via Faster Enumeration}
\label{subsec:equiv_hyperclique}

What is left to do is to connect the decision problems regarding
the transversal rank/conformal degree with the enumeration problems
of minimal hitting sets/maximal hypercliques in both directions.
We first assume access to an algorithm that enumerates the hypercliques
of an $r$-uniform hypergraph in iterative-polynomial time $\poly(i) \cdot n^{r+O(1)}$.
We derive from it an $\poly(m) \cdot n^{k+O(1)}$ algorithm
to solve \ConfDeg for \emph{non-uniform} input hypergraphs $\Hyp$ and arbitrary parameter $k$.

Let $[\Hyp]_{k} = \{ S \in \binom{V}{k} \mid \exists E \in \Hyp \colon S \subseteq E\}$
be the $k$-\emph{section} of $\Hyp$,
the hypergraph that contains as edges the $k$-sets of vertices that appear together in an edge of $\Hyp$.

\begin{lemma}
\label{lem:conformality}
	The hypergraph $\Hyp$ is $k$-conformal if and only if $\C([\Hyp]_k) \subseteq \Hyp$.
\end{lemma}

\begin{proof}
	Let $C \in \C([\Hyp]_k)$ be a maximal hyperclique of the $k$-section.
	That means, each of its subsets of size at most $k$ is contained in some hyperedge.
	If $\Hyp$ is $k$-conformal, then $C$ itself is contained in some edge $E \in \Hyp$.
	Furthermore, note that $E$ is also a hyperclique of $[\Hyp]_k$.
	By the maximality of $C$, we get $C = E \in \Hyp$.
	Since $C$ is arbitrary, we have $\C([\Hyp]_k) \subseteq \Hyp$.
	
	Conversely, let $S$ be a set such that each of its subsets of size at most $k$
	is contained in an edge of $\Hyp$.
	In particular, $S$ is a hyperclique in $[\Hyp]_k$.
	Let $C$ be a maximal hyperclique with $C \supseteq S$.
	If $\C([\Hyp]_k) \subseteq  \Hyp$ indeed holds,
	then $C$ is an edge of $\Hyp$ that contains $S$.
	In other words, $\Hyp$ is $k$-conformal.
\end{proof}
We point out that \autoref{lem:conformality} can easily be strengthened to show 
that $k$-conformality is equivalent to $\C([\Hyp]_k)$ containing precisely 
the inclusion-wise maximal edges of $\Hyp$.
However, we only need the weaker version.

The lemma suggests the following algorithm to certify that $\Hyp$ is $k$-conformal
or produce a counterexample $S$.
We first compute the $k$-section of $\Hyp$ in time 
$O(km \cdot \binom{n}{k}) = O(mn^{k})$.
(Note that $\binom{n}{k} = O(n^k/ \nwspace k!)$,
hence the above estimate holds even if $k$ is not constant.)
Clearly, $[\Hyp]_k$ is $k$-uniform.
The assumed hyperclique enumeration algorithm lists the $i$-th member of $\C([\Hyp]_k)$ 
in time $ i^c \cdot n^{k+O(1)}$ for some constant $c$.
For each hyperclique $C$, we check whether $C \in \Hyp$ is an edge.
We defer the description how to perform this check in time $O(n^2)$ until the end of this section.
The enumeration can be aborted after at most $m+1$ hypercliques
as the $(m{+}1)$-st one already certifies $\C([\Hyp]_k) \nsubseteq \Hyp$.

If, at any point, we find a hyperclique $C \in \C([\Hyp]_k){\setminus}\Hyp$,
then, evidently, every subset of $C$ of cardinality at most $k$ is contained in some edge of $\Hyp$.
We claim that $C$ itself is not contained in an edge and is 
therefore a valid output for the search problem \ConfDeg.
By definition, $C$ it is not \emph{equal} to any edge in $\Hyp$.
If it were \emph{strictly} contained in some $E \in \Hyp$,
then $E$ would be a hyperclique in $[\Hyp]_k$ extending $C$,
which is a contradiction to the maximality of $C$.
The algorithm takes total time
\begin{gather*}
	 O(mn^{k}) + \sum_{i=0}^{m+1} i^c \cdot n^{k+O(1)} =
	 	(m{+}1)^{c+1} \cdot n^{k+O(1)} = \poly(m) \cdot n^{k+O(1)}.
\end{gather*}

To support the check whether the clique $C$ is in $\Hyp$,
we spend $O(mn)$ time to preprocess the hypergraph into a dictionary
allowing membership queries in time $O(|C| \cdot \log |\Hyp|) = O(n \log m)$.
Since any hypergraph has at most $m = O(2^n)$ edges,
the query time is $O(n^2)$.

\subsection{Faster Enumeration via Faster Decision}
\label{subsec:equiv_hitting_sets}

We now show the inverse direction of the equivalence.
We assume that  there exists a constant $c \ge 1$ such that
\TransRank is decidable in time $m^c \cdot n^{k+O(1)}$.
We claim that, under this assumption,
the minimal hitting sets of any hypergraph $\Hyp$ with rank $r$ can be enumerated in such a way
that the $i$-th delay is bounded by $i^{c} \cdot n^{r+O(1)}$.
Eiter and Gottlob~\cite{EiterGottlob95RelatedProblems},
were the first to show that this problem admits an incremental-polynomial algorithm
(albeit with a much worse delay).
They used repeated calls to a decision
subroutine to check whether $\Gyp = \Tr(\Hyp)$ holds for a certain hypergraph $\Gyp$.
The subroutine runs in time $|\Gyp|^{r+2} \cdot \poly(n,|\Hyp|)$
using the lemma below.\footnote{%
	Strictly speaking, Eiter and Gottlob~\cite{EiterGottlob95RelatedProblems}
	prove \autoref{lem:EiterGottlob} only for the case where both $\Gyp$ and $\Hyp$ are Sperner.
	We give a short proof of the general case in \autoref{app:EiterGottlob} for completeness.
}
The main issue is that this algorithm
is too slow for our purposes.
To improve the delay, we need to reduce the dependency on the size of $\Gyp$ to $|\Gyp|^c$
(i.e, make the exponent independent of $r$).
We are willing to accept an $n^{r+O(1)}$ dependency instead.

In the following, we let $\Tr(\Gyp)|_r$ denote the collection of all minimal hitting sets of $\Gyp$
with cardinality at most $r$.
Evidently, $\rank(\Tr(\Gyp)) \le r$ holds if and only if $\Tr(\Gyp)|_r = \Tr(\Gyp)$.

\begin{restatable}{lemma}{EiterGottlob}
\textup{\textsf{(Eiter and Gottlob~\cite{EiterGottlob95RelatedProblems})}}
\label{lem:EiterGottlob}
	Let $\Gyp,\Hyp$ be two hypergraphs such that $r = \rank(\Hyp)$.
	Then, $\Gyp = \Tr(\Hyp)$ is true if and only if the following three conditions are fulfilled.
	\vspace*{.25em}
	\begin{enumerate}
		\item \label[cond]{cond:subset} $\Gyp \subseteq \Tr(\Hyp)$ - all edges of $\Gyp$ are minimal hitting sets for $\Hyp$;
		\vspace*{.25em}
		\item \label[cond]{cond:bounded_HS} $\Tr(\Gyp)|_{r} \subseteq \Hyp$ - all minimal hitting sets for $\Gyp$ with at most $r$ vertices are edges of $\Hyp$;
		\vspace*{.25em}
		\item \label[cond]{cond:bounded_TR} $\rank(\Tr(\Gyp)) \le r$ - the transversal rank of $\Gyp$ is at most $r$.
	\end{enumerate}
\end{restatable}
\vspace*{.5em}
 
Checking whether a set $T \subseteq V$ of vertices is a minimal hitting set of $\Hyp$ 
can be done in time $O(|\Hyp||T|)$.
Hence, the first condition is testable in time $|\Gyp| \cdot O(|\Hyp| \nwspace n)$.
Regarding \Cref{cond:bounded_HS},
it is enough to brute force all subsets of $V$ with at most $r$ vertices,
check whether they are minimal hitting sets of $\Gyp$ and, if so,
whether they are also an edge of $\Hyp$.
Moreover, the latter test is triggered at most $|\Hyp|$ times
as $|\Tr(\Gyp)|_r| > |\Hyp|$ is sufficient for $\Tr(\Gyp)|_{r} \nsubseteq \Hyp$.
The whole verification takes time
$O(|\Gyp| \nwspace r \binom{n}{r} + |\Hyp| \nwspace r)) = |\Gyp| \cdot O(\nwspace n^{r})$.
The latter estimate uses that $\Hyp$ is an $r$-hypergraph,
which implies $|\Hyp| \le \binom{n}{\le r} = O(\binom{n}{r}) = O(n^r/ \nwspace r!)$.

The problem is with \Cref{cond:bounded_TR}.
Eiter and Gottlob~\cite{EiterGottlob95RelatedProblems} use the result by Berge and Duchet~\cite{BergeDuchet75GilmoresTheorem} to decide it in time
$O(|\Gyp|^{r+2} \nwspace n)$,
we use the improved algorithm instead.
Evidently, $\Gyp$ has transversal rank at most $r$ if and only if it is a \emph{no}-instance
for the \TransRank problem with parameter $r+1$.
We check this in time $|\Gyp|^c \cdot n^{r+O(1)}$.
In total, this gives the desired method to decide $\Gyp = \Tr(\Hyp)$
in time $ |\Gyp|^c \cdot n^{r+O(1)}$.
(See also \autoref{thm:second_equivalence}.)

We turn this into an enumeration algorithm for $\Tr(\Hyp)$.
It maintains a set of solutions $\Gyp \subseteq \Tr(\Hyp)$
and proceeds in stages numbered from $0$ to $|\Tr(\Hyp)|$.
The $i$-th stage corresponds to the number $|\Gyp|$ of minimal hitting sets found so far.
In each stage, the algorithm checks whether $\Gyp = \Tr(\Hyp)$ already holds.
If so, it terminates.
Otherwise, it has to find a new minimal solution in $\Tr(\Hyp){\setminus}\Gyp$.
By \autoref{lem:EiterGottlob}, 
it must be that $\Tr(\Gyp)|_r \nsubseteq \Hyp$ or $\rank(\Tr(\Gyp)) > r$.

First, suppose we have $\Tr(\Gyp)|_r \nsubseteq \Hyp$,
hence there exists a minimal hitting set $S$ for $\Gyp$ with $|S| \le r$ vertices,
which is not an edge of $\Hyp$.
The invariance $\Gyp \subseteq \Tr(\Hyp)$
states that any edge of $\Gyp$ intersects each one in $\Hyp$.
Conversely, it also means that each edge\footnote{
	In this proof, we use the symbol $H$ for edges of $\Hyp$ to better distinguish them from the edges $G \in \Gyp$.
}
$H$ in $\Hyp$ is a hitting set for $\Gyp$.
Since $S$ is minimal,
if there were some $H \in \Hyp$ with $H \subseteq S$,
we would have $H = S$ and thus $S \in \Hyp$, a contradiction.
It follows that $H \nsubseteq S$ for each $H \in \Hyp$.
This is equivalent to the complement $V{\setminus}S$ being a (not necessarily minimal)
hitting set for $\Hyp$.
Since $S \in \Tr(\Gyp)|_r$ is a hitting set for $\Gyp$,
no edge $G \in \Gyp$ is entirely contained in $V{\setminus}S$.
Therefore, \emph{any} minimal hitting set $T \subseteq V{\setminus}S$
lies in $\Tr(\Hyp){\setminus}\Gyp$.
Such a new solution $T$ can be computed from $S$
in time $O(|\Hyp| \nwspace n) = n^{r+O(1)}$.

Observe that the only two properties of $S$ we used
were that $S \in \Tr(\Gyp)$ is a minimal hitting set for $\Gyp$,
and that it is not an edge of $\Hyp$.
In the other case where $\rank(\Tr(\Gyp)) > r$, 
the algorithm for \TransRank with parameter $r+1$ provides
a minimal hitting set $S$ for $\Gyp$ with $|S| > r$.
The rank of $\Hyp$ is $r$, hence $S \notin \Hyp$.
The same argument as above again shows that any minimal hitting set $T$ for $\Hyp$
with $T \subseteq V{\setminus}S$ lies in $\Tr(\Hyp){\setminus}\Gyp$.
In summary, computing a new solution 
is dominated by the $i^c \cdot n^{r+O(1)}$ time to verify $\Gyp \neq \Tr(\Hyp)$.

\subsection{Proofs of \texorpdfstring{\Cref{thm:second_equivalence,thm:conformality_test}}
	{Theorems 5 and 6}}
\label{subsec:equiv_remarks}

The results above also imply an equivalence between an improved enumeration algorithm
and a faster way to decide $\Gyp = \Tr(\Hyp)$.

\secondequivalence*

\begin{proof}
	One direction has been shown already in the first part of \autoref{subsec:equiv_hitting_sets}.
	The other one can be obtained in a similar fashion as the second part above.
	Assume access to a black-box procedure that decides whether $\Gyp = \Tr(\Hyp)$ holds,
	and otherwise returns some edge in 
	$(\Gyp{\setminus} \Tr(\Hyp)) \cup (\Tr(\Hyp){\setminus}\Gyp)$.
	Then, starting with $\Gyp = \emptyset \subseteq \Tr(\Hyp)$
	and maintaining $\Gyp \subseteq \Tr(\Hyp)$ ensures that
	we always get a new edge from $\Tr(\Hyp){\setminus}\Gyp$ until $\Gyp = \Tr(\Hyp)$ is reached.
	If the black-box algorithm runs in time $|\Gyp|^c \cdot n^{r+O(1)}$ for $r$-hypergraphs $\Hyp$,
	the $i$-th delay of the enumeration is $i^c \cdot n^{r+O(1)}$.
\end{proof}

\autoref{thm:main_equivalence} shows that we need better enumeration algorithms for maximal hypercliques
in $r$-uniform hypergraphs to obtain improved algorithms for the \TransRank problem and thus \ConfDeg.
However, for the simplest case $r = 2$,
that is, for maximal cliques in ordinary graphs,
very fast enumeration algorithms are already known.
This gives an improved conformality test.

\conformalitytest*

\begin{proof}
	The maximal cliques of a graph $G = (V,E)$ can be enumerated
	with delay $O(|E|n) = O(n^3)$, independently of the number of solutions 
	seen so far~\cite{Johnson88MaxIndSet,Tsukiyama77GeneratingAllTheMaxIndSets}.
	We plug this into the method in \autoref{subsec:equiv_hyperclique} 
	to decide the conformality of a hypergraph,
	that is, \ConfDeg with $k = 2$.
	In more detail, we construct the 2-section $[\Hyp]_2$
	(a.k.a.\ the co-occurrence graph or clique expansion) 
	that has an edge between any pair of vertices in $V$ that appear together in an edge of $\Hyp$.
	Enumerating the first $m = |\Hyp|$ of its maximal cliques takes time $O(mn^3)$.
	If the enumeration algorithm does not terminate after the $m$-th clique,
	or if one of them is not an edge of $\Hyp$,
	then $\Hyp$ is not conformal; otherwise it is.
\end{proof}

\bibliography{enum_refs}

\appendix

\section{Proofs Omitted from the Introduction}
\label{app:intro_proofs}

\noindent
\textbf{\TransRank in time $O(m^{k+1} + m^k n)$.}
For a hypergraph $\Hyp$,
let $\min(\Hyp)$ be the \emph{minimization} of $\Hyp$,
that is, the subhypergraph of its inclusion-wise minimal edges.
Note that $\min(\Hyp)$ is Sperner for any $\Hyp$.
It does not make a difference for the hitting sets
whether the minimization or the full hypergraph is considered,
we have $\Tr(\Hyp) = \Tr(\min(\Hyp))$.
Moreover, the minimal edges are precisely
the minimal hitting set of the hitting sets, $\min(\Hyp) = \Tr(\Tr(\Hyp))$,
see \cite{Berge89Hypergraphs}.

Without loss of generality, we can assume that $k \ge 3$.
This implies that the $O(m^2 \nwspace n)$ time needed to compute $\min(\Hyp)$
is immaterial compared to the running time of \TransRank.
We thus assume that the input hypergraph $\Hyp$ is Sperner.
For some $k$-edge subhypergraph $\Gyp = \{E_1, E_2, \dots, E_k\} \subseteq \Hyp$,
we let $D(\Gyp) = \bigcup_{1 \le i < j \le k} (E_i \cap E_j)$ denote the set of all vertices
that appear in at least two edges of $\Gyp$.
With this notation, \autoref{lem:BergeDuchet_subhypergraph} states 
that $\rank(\Tr(\Hyp)) \ge k$ is equivalent to the existence of 
a $k$-subhypergraph such that $S = V{\setminus}D(\Gyp)$ is a hitting set for $\Hyp$.
The latter, in turn, is equivalent to $D(\Gyp)$ not containing any edge of $\Hyp$.
On a high level, our algorithm searches for $k$ subhypergraphs of 
$\Fyp_i \subseteq \Hyp$, each with $k-1$ edges,
that can be combined to some $\Gyp$ such that $E \nsubseteq D(\Gyp)$ holds for all $E \in \Hyp$.
	
We first give an alternative way to compute the set $D(\Gyp)$, namely
\begin{equation}
\label{eq:reformulation}
	D(\Gyp) = \bigcap_{i=1}^k 
		\left( E_1 \cup \dots \cup E_{i-1} \cup E_{i+1} \cup \dots \cup E_k \right). 
\end{equation}
To see this equality, note that any vertex $v \in V$ that is in at least two different edges of $\Gyp$
is also contained in any union of all but one edge of $\Gyp$.
Conversely, suppose that the degree of $v$ in $\Gyp$ is at most $1$
and let $E_i \in \Gyp$ be the unique edge containing $v$,
or $E_i = E_1$ if no such edge exists.
Then, $v \notin E_1 \cup \dots \cup E_{i-1} \cup E_{i+1} \cup \dots \cup E_k$.

In the following, we assume an (arbitrary) order among the edges of $\Hyp$.
Fix a subhypergraph $\Fyp \subseteq \Hyp$ with $k{-}1$ edges.
The ordered list $L(\Fyp)$ contains all edges $E \in \Hyp$ for which 
$E \subseteq \bigcup \Fyp$, in the order induced by $\Hyp$.
The ordered lists for \emph{all} possible $(k{-}1)$-hypergraphs
can be computed in total time $\binom{m}{k-1}\cdot \big(O(kn) + O(mn) \big) = O(m^k n)$.

For some $\Gyp \subseteq \Hyp$ with $k$ edges,
let $\Fyp_1, \dots, \Fyp_k \subseteq \Gyp$ be all its subhypergraphs
that are missing exactly one edge.
By \autoref{eq:reformulation}, we have $E \subseteq D(\Gyp)$ for some $E \in \Hyp$
if and only if $E$ appears on every list $L(\Fyp_1), \dots, L(\Fyp_k)$.
Equivalently, if $L(\Fyp_1) \cap \dots \cap L(\Fyp_k) = \emptyset$,
then $\Gyp$ witnesses that $\rank(\Tr(\Hyp)) \ge k$.

Since the lists are ordered and have at most $m$ entries,
we can intersect them in time $O(km)$.
Doing this for all $k$-edge subhypergraphs of the input $\Hyp$
takes total time $\binom{m}{k} \cdot O(km) = O(m^{k+1})$.
Note that the last bound does not assume $k$ to be constant.
If the search finds a suitable subhypergraph $\Gyp$,
then \emph{any} minimal hitting set contained in $V{\setminus}D(\Gyp)$
has size at least $k$.
We show below how to compute one of them in time $O(mn)$.
\vspace*{.5em}

\noindent
\textbf{Minimizing a hitting set.}
We occasionally have the situation that we are given some hitting set $S$ for some hypergraph $\Hyp$
and need to compute an arbitrary \emph{minimal} hitting set $T \subseteq S$.
The standard greedy algorithm tries all vertices $s \in S$ and 
tests whether $S{\setminus}\{s\}$ is still a hitting set;
if so, $s$ is removed.
The result is guaranteed to be a minimal hitting set.
The bottleneck are the repeated tests whether all edges of $\Hyp$ are hit,
resulting in a total running time of $O(m \nwspace |S|^2)$.
We now show how to reduce this to $O(m \nwspace |S|)$.

The algorithm first takes $O(m \nwspace |S|)$ time to build the bipartite incidence graph of $S$.
In more detail, that is the graph $I$ with vertices $S \cup \Hyp$ 
such that $s \sim E$ is an edge iff $s \in E$.
Since $S$ is a hitting set, the degree $\deg_I(E)$ of any hyperedge $E \in \Hyp$ is at least $1$.
The algorithm constructs the set $T$ as follows.
It does a first pass over the hyperedges $E \in \Hyp$ and, if $\deg_I(E) = 1$,
the unique neighbor $s \in N_I(E)$ is put in $T$ and 
its closed neighborhood $N_I[s]$, (that is, $s$ and all hyperedges containing $s$)
is removed from the incidence graph.

If this does not already exhaust the partition $\Hyp$, 
then all remaining hyperedges have degree at least $2$.
In the second phase, the algorithm takes an arbitrary\footnote{
	Different tie breaking rules for the hyperedge and its neighbor
	may lead to different minimal sets $T \subseteq S$.
	We ignore these implementation details since
	we are only looking for \emph{some} minimal hitting set.
} remaining $E \in \Hyp$
and one of its neighbors $s$ and removes $s$ from the graph
(without putting it in $T$ or removing any hyperedges).
This cannot isolate any member of $\Hyp$,
but may result in some of them having degree $1$.
If so, the unique neighbor is included in $T$ in the same fashion as described above,
reducing the number of remaining hyperedges.
The procedure ends once all hyperedges have been removed that way.

The set $T$ is a hitting set for $\Hyp$, because a hyperedge only gets removed
if one of its vertices is included in $T$.
The degree-1-rule further ensures that every $s \in T$ has a private hyperedge, thus $T$ is minimal.
The key observation for the running time is that each vertex $s \in S$
can be processed in time $O(\deg_I(s))$,
regardless of whether it is included in $T$ or not.
More precisely, there exists a universal constant $c > 0$ such that the procedure takes time
\begin{equation*}
	O(m \nwspace |S|) + \sum_{s \in S} c \cdot \deg_I(s) = O(m \nwspace |S|).
\end{equation*}

\section{Proof of \autoref{lem:EiterGottlob}}
\label{app:EiterGottlob}

We consider a certain binary relation over hypergraphs.
We write $\Gyp \preccurlyeq \Hyp$ if every edge of $\Gyp$ 
contains an edge of $\Hyp$ as a subset.
Note that $\preccurlyeq$ is a pre-order in that it is reflexive and transitive.
Moreover, $\Gyp \subseteq \Hyp$ implies $\Gyp \preccurlyeq \Hyp$.
When restricted to Sperner hypergraphs, the relation is even antisymmetric
(i.e., a partial order).
The key fact we use is that the relation is inverted when taking hitting sets
as stated in the following result by Birnick, Bläsius, Friedrich, Naumann, Papenbrock and Schirneck~\cite{Birnick20HPIValid}

\begin{lemma}[Birnick et al.~\cite{Birnick20HPIValid}]
\label{lem:duality_transversal}
	Let $\Gyp$ and $\Hyp$ be hypergraphs.
	Then, $\Gyp \preccurlyeq \Hyp$ holds iff $\Tr(\Gyp) \succcurlyeq \Tr(\Hyp)$.
\end{lemma}


\EiterGottlob*

\begin{proof}
	Recall that $\Tr(\Tr(\Hyp) = \min(\Hyp)$ holds for the minimization of $\Hyp$
	(see \autoref{app:intro_proofs}).
	Suppose we indeed have $\Gyp = \Tr(\Hyp)$.
	Then, \Cref{cond:subset} is not in question.
	The assumption further implies
	$\Tr(\Gyp)|_r = \Tr(\Tr(\Hyp))|_r = \min(\Hyp)|_r \subseteq \min(\Hyp) \subseteq \Hyp$
	as well as 
	$\rank(\Tr(\Gyp)) = \rank(\min(\Hyp)) \le \rank(\Hyp) = r$,
	which are the last two conditions.
	
	For the opposite direction, 
	we claim that if all three conditions are fulfilled, we have $\min(\Gyp) = \Tr(\Hyp)$.	
	The inclusion $\Gyp \subseteq \Tr(\Hyp)$ 
	implies that the hypergraph $\Gyp$ is Sperner, that is, $\min(\Gyp) = \Gyp$.
	Establishing the claim is therefore sufficient for the desired $\Gyp = \Tr(\Hyp)$.
	
	Regarding the claim, we have $\min(\Gyp) \subseteq \Gyp \subseteq \Tr(\Hyp)$
	and thus $\min(\Gyp) \preccurlyeq \Tr(\Hyp)$.
	\Cref{cond:bounded_HS,cond:bounded_TR} together imply $\Tr(\Gyp) \subseteq \Hyp$,
	and therefore $\Tr(\Gyp) \preccurlyeq \Hyp$.
	Applying \autoref{lem:duality_transversal} to $\Tr(\Gyp)$ and $\Hyp$ gives
	$\min(\Gyp)  = \Tr(\Tr(\Gyp)) \succcurlyeq \Tr(\Hyp)$.
	Since $\min(\Gyp)$ and $\Tr(\Hyp)$ are both Sperner, the antisymmetry of the relation $\preccurlyeq$ 
	gives $\min(\Gyp) = \Tr(\Hyp)$
\end{proof}

\end{document}